\documentclass[journal]{IEEEtran}
\usepackage{amsmath,amssymb,amsfonts,bm}
\usepackage{algorithmic}
\usepackage{graphicx}
\usepackage{float}
\usepackage{mdwmath}
\usepackage{booktabs}

\usepackage{amssymb}
\usepackage{amsmath}
\usepackage{amsthm}
\usepackage{mathrsfs}

\usepackage{textcomp}
\usepackage{subfigure}
\usepackage{float}
\usepackage{subfig}

\usepackage{soul,color}
\usepackage{cancel}
\usepackage{cite}

\theoremstyle{plain}
\newtheorem{thm}{Theorem}

\newtheorem{lem}[thm]{Lemma}
\newtheorem{prop}[thm]{Proposition}

\newtheorem{defn}[thm]{Definition}

\newcommand{\norm}[1]{\left\lVert#1\right\rVert}

\ifCLASSINFOpdf

\else

\fi

\begin{document}

\title{Target Detection within Nonhomogeneous Clutter via Total Bregman Divergence-Based Matrix Information Geometry Detectors}

\author{Xiaoqiang~Hua, Yusuke Ono, Linyu~Peng, Yongqiang~Cheng, Hongqiang~Wang 

\thanks{This work was partially supported by the National Natural Science Foundation of China under Grant Number 61901479, JSPS KAKENHI Grant Number JP20K14365,  JST-CREST Grant Number JPMJCR1914, and Keio Gijuku Fukuzawa Memorial Fund. {\it (Corresponding author: Linyu Peng.)}

X. Hua is with the College of Meteorology and Oceanography, National University of Defense Technology, Changsha, Hunan 410073, China (e-mail: hxq712@yeah.net).

Y. Ono is with the Department of Mechanical Engineering, Keio University, Hiyoshi 3-14-1, Yokohama 223-8522, Japan (e-mail: yuu555yuu@keio.jp).

L. Peng is with the Department of Mechanical Engineering, Keio University, Hiyoshi 3-14-1, Yokohama 223-8522, Japan. He  is also an adjunct faculty member of Waseda Institute for Advanced Study, Waseda University, Japan, and School of Mathematics and Statistics, Beijing Institute of Technology, China (e-mail:
 l.peng@mech.keio.ac.jp).

Y. Cheng and H. Wang are with the College of Electronic Science, National University of Defense Technology, Changsha, Hunan 410073, China (e-mail: cyq101600@126.com; oliverwhq@tom.com).}}



\maketitle

\begin{abstract}
Information divergences are commonly used to measure the dissimilarity of two elements on a statistical manifold. Differentiable manifolds endowed with different divergences may possess different geometric properties, which can result in totally different performances in many practical applications. In this paper, we propose a total Bregman divergence-based matrix information geometry (TBD-MIG) detector and apply it to detect targets emerged  into nonhomogeneous clutter. In particular, each sample data is assumed to be modeled as a Hermitian positive-definite (HPD) matrix and the clutter covariance matrix is estimated by the TBD mean of a set of secondary HPD matrices. We then reformulate the problem of signal detection as discriminating two points on the HPD matrix manifold. Three TBD-MIG detectors, referred to as the total square loss, the total log-determinant and the total von Neumann MIG detectors, are proposed, and they can achieve great performances due to their power of discrimination and robustness to interferences. Simulations show the advantage of the proposed TBD-MIG detectors in comparison with the geometric detector using an affine invariant Riemannian metric  as well as the adaptive matched filter  in nonhomogeneous clutter.
\end{abstract}


\IEEEpeerreviewmaketitle

\section{Introduction}

\IEEEPARstart{D}{etecting} a target of interest emerging into nonhomogeneous clutter is always a challenging subject in fields of radar, sonar and communications. Typically,  detection performance is mainly affected by the estimate accuracy of clutter covariance matrix (CCM) \cite{839973}. Classical sample covariance matrix (SCM) estimators are derived using a set of secondary data collected from range gates spatially close to the one under test according to the maximum likelihood estimation (MLE) criterion \cite{9097436,8052571}, and have been widely used in the generalized likelihood ratio test (GLRT) detectors \cite{4104190,135446,511809}, and Rao and Wald tests \cite{7384511,7605536,8968391}. However, performances of SCM estimators are sensitive to nonhomogeneous clutter due to the limited number of homogeneous sample data as well as the heterogeneity. On the one hand, sufficient number of homogeneous sample data that is independent and identically distributed  and shares the same spectral property is needed to achieve a satisfactory estimate performance. For instance, to guarantee the performance loss less than $3$ dB, the number of homogeneous sample data must be more than $2$ times of the dimension of sample data. On the other hand, the sample data is inevitably contaminated by outliers caused by the interferences or the variation of clutter power, although the outliers can be censored by the sample selection methods,  such as, the training sample selection with matrix whitening \cite{9067074}, the geometric mean or median-based generalized inner product (GIP) \cite{6573681,6825699}, the covariance structure selection \cite{CHEN199910} and the MLE method \cite{1039406}. An acceptable performance cannot be obtained unless  sufficient number of homogeneous sample data are available. However, the number of homogeneous sample data can be very limited in real nonhomogeneous clutter, that often results in a remarkable degradation in detection performance.

An effective strategy to improve  detection performance is to incorporate {\it a priori} information about the nonhomogeneous clutter environment into the detector design, namely to perform a knowledge-based processing. For instance, in \cite{4267625}, the environmental information provided by geographic information system is employed to select the homogeneous sample data, and a significant improvement in the detection performance is achieved in real IPIX radar data. In \cite{5417154}, the unknown CCM is assumed to obey the complex Wishart and inverse complex Wishart distributions, and two GLRT-based detectors are designed in a Bayesian framework. The advantage of the proposed detectors with respect to their non-Bayesian counterparts is validated on real L-band clutter data. Another example is provided in \cite{6020816}, where the CCM is modeled as a multi-channel auto-regressive process. Based on this model, two knowledge-aided parametric adaptive detectors are derived by employing a prior information about the spatial correlation through colored-loading. The performance analysis on various datasets reveals the advantage of the proposed parametric adaptive detectors, especially in the case of limited data (see also \cite{8289377,5484507,4154721}). That knowledge-aided target detection method can achieve  performance improvement in nonhomogeneous clutter is mainly due to that statistical characteristics of clutter environment are provided. Unfortunately, statistical characteristics of clutter are often unknown or are difficult to capture in real practical applications. Insufficient knowledge about clutter often leads to a severe performance degradation as well.

Another approach to circumvent the degradation in detection performance is to design the detector in the framework of matrix information geometry (MIG). This type of detectors do not require {\it a priori}  knowledge about the statistical characteristics of clutter environment but simply invoke  the geometry of Riemannian manifolds. MIG is a relatively new branch in the study of information geometry, which was pioneered by Rao in the 1940s \cite{Rao1992} and further developed by Chentsov \cite{Chentsov1972}, Efron \cite{Efron1975Defining},   Amari \cite{97844712097,amari2000methods}, etc. Information geometry  studies intrinsic properties of statistical models, and many  information processing problems from information science can be transformed into discriminational problems on differentiable manifolds equipped with a Riemannian metric or, in particular, an affine invariant Riemannian metric (AIRM). MIG is a natural  extension of classical information geometry.  Lots of signal processing problems have been successfully solved using the MIG theory. For instance, in \cite{6573681,7887259}, the CCM estimation related to a geometric distance is transformed into computing the geometric barycenter of basic covariance matrices. The basic covariance matrices are constructed by a set of secondary data with a condition number upper bound constraint. Then, the GIP is used together with the estimated covariance matrix to design a training sample selector. The results have shown significant performance improvements over the GIP method in nonhomogeneous clutter. In \cite{7346232}, a signal detection method based on the Riemannian $p$-mean of covariance matrices estimated by a neighbourhood of the considered cell is designed on the Toeplitz Hermitian positive-definite (HPD)  matrix Riemannian manifold. This detector is called the matrix constant false alarm rate (CFAR) detector or MIG detector. The advantage of MIG detector has been shown on target detection in high frequency X-band radar clutter \cite{4721049}, Burg estimation of Radar scatter matrix \cite{7842633}, the analysis of statistical characterization \cite{9078971} and monitoring of wake vortex turbulences \cite{Liu2013,BARBARESCO201054}.  It is worth noting that the diagonal loading is also an useful tool for CCM estimation. This can be achieved by resorting to adding the identity multiplied by a loading factor to the SCM. Diagonal loading has been successfully applied to target detection \cite{4383590,8450037} and adaptive beamforming \cite{7181,5417174}. For instance, diagonal loading can be used to reduce the main-lobe distortion maintaining also lower sidelobes to stabilize the beampattern, as the large eigenvalues of the CCM due to strong interference are not significantly affected by the loading process, whereas the smaller eigenvalues are increased \cite{8450037}. Recently, a novel MIG detector based upon  information divergences that possess many nice properties rather than the geodesic distance is proposed \cite{HUA2017106,8000811,HUA2019640}. Specifically, the sample data is assumed to be modeled as an HPD matrix. The new observation is that an HPD matrix represents the power or correlation of the sample data. The set of all HPD matrices form a differentiable  manifold with non-positive curvature \cite{bridson2013metric}. Then, the problem of target detection can be treated as discriminating two points on the HPD matrix manifold. The CCM is estimated as the geometric mean of a set of secondary HPD matrices. Since the geometric mean is robust to outliers and the detector does not rely on {\it a priori} knowledge about clutter environment, the new detector leads to better performance  over the conventional detector in nonhomogeneous clutter.

In this paper, we extend our previous ideas, presented in \cite{HUA2018232}, by proposing a class of total Bregman divergences (TBDs) on the HPD matrix manifold, and designing a TBD-based MIG (named as TBD-MIG) detector for target detection in nonhomogeneous clutter. Contributions  in this paper are summarized as follows.

\begin{enumerate}
  \item TBD on the HPD matrix manifold is defined, motivated by the TBD defined on vector spaces. Specifically, three TBDs, including the total square loss (TSL), the total von-Neumann (TVN) divergence, and the total log-determinant (TLD) divergence, are defined by resorting to different convex functions. We also analyze  the difference in geometric structures of TBDs on the HPD matrix manifold. Several computational confusions in \cite{HUA2018232} are also clarified in the current paper.
  \item Geometric mean associated with the TBD for a set of HPD matrices is defined. We derive the TSL, TVN and TLD means in closed-form using the stationary condition on the HPD matrix manifold and use them as the estimators of CCM. In addition, an influence function that describes the influence of the outlier on the TBD mean is proposed to analyze  the robustness. The results show that the TBD mean is much more robust to the strong outlier compared with the AIRM mean.
  \item Experiments performed on simulation database verify the advantages of the proposed TBD-MIG detector compared with the AIRM-MIG detector and the adaptive matched filter (AMF). Moreover, we analyze  the influence of different matrix structures on detection performance.
\end{enumerate}

The rest of the paper is organized as follows. Section II reformulates the problem of signal detection on the HPD matrix manifold and describes the framework of MIG detector. Section III provides a brief mathematical knowledge of MIG. In Section IV, we define the TBD on the HPD matrix manifold and deduce the TBD mean using the stationary condition. Influence functions are derived in closed-form and the analysis of robustness to the outliers is shown numerically. Section V provides the simulation results and comparative analysis using the proposed detector. Conclusions are summarized in Section VI.

 \textit{Notations:} In the sequel, scalars, vectors and matrices are denoted by lowercase, boldface lowercase and boldface uppercase letters, respectively. The symbols $(\cdot)^{\operatorname{T}}$ and $(\cdot)^{\operatorname{H}}$ stand for the transpose and conjugate transpose of matrices, respectively. The operators $\operatorname{det}(\cdot)$ and $\operatorname{tr}(\cdot)$ denote the determinate and trace of a matrix. The  $N\times N$ identity matrix is denoted by  $\bm{I}_N$ or simply $\bm{I}$ if no confusion would be caused.  The Frobenius norm of a matrix with respect to the Frobenius metric is simply denoted by $\norm{ \cdot }$. We use $\mathbb{C}^N$ to represent the set of $N$-dimensional complex vectors. Finally, $\operatorname{E}[\cdot]$ denotes the statistical expectation.

\section{Problem Formulation}
Suppose that the sample data are collected from $N$ (temporal, spatial, or spatial-temporal) channels. We consider the problem of detecting a moving target embedded in clutter. In general, the detection problem can be formulated as the following binary hypothesis  testing, namely
\begin{equation}
\left\{
\begin{aligned}
&\mathcal{H}_0: \left\{
\begin{aligned}
&\bm{x} = \bm{c} \\
&\bm{x}_k = \bm{c}_k, \quad k=1,2,\ldots,K,
\end{aligned} \right.\\
&\mathcal{H}_1: \left\{
\begin{aligned}
&\bm{x} = \alpha \bm{p} + \bm{c}, \\
&\bm{x}_k = \bm{c}_k, \quad k=1,2,\ldots,K,
\end{aligned} \right. \\
\end{aligned} \right.
\end{equation}
where $\alpha$ is  unknown and complex scalar-valued,  accounting for the channel propagation effects and target reflectivity, the vectors $\bm{c}$ and $\bm{c}_k, k=1,2,\ldots,K$ denote the clutter data, $\bm{x}$ and $\bm{x}_k, k=1,2,\ldots,K$ stand for the sample data, and  $\bm{p}$ denotes the known signal steering vector. Write column vectors
\begin{equation}
\begin{aligned}
&\bm{x} = [x_0, x_1, \ldots, x_{N-1}]^{\operatorname{T}}\in \mathbb{C}^N, \\
&\bm{p} = \frac{1}{\sqrt{N}}[1, \exp(-\operatorname{i}2\pi f_d), \ldots, \exp(-\operatorname{i}2\pi f_d(N-1))]^{\operatorname{T}},
\end{aligned}
\end{equation}
where $f_d$ denotes the normalized Doppler frequency, and $\operatorname{i}$ is the imaginary unit. 

The correlation or power of the sample data is considered for discriminating the target signal and clutter. Important features of the sample data can be captured by a special HPD matrix with either  the Toeplitz structure or using the diagonal loading method.

For the sample data $\bm{x}$, the Toeplitz HPD feature matrix is
\begin{equation}
\begin{aligned}
\bm{R} = \operatorname{E}[\bm{x} \bm{x}^{\operatorname{H}}] =
    \begin{bmatrix}
    r_0    & \bar{r}_1 & \cdots & \bar{r}_{N-1} \\
    r_1    &  r_0      & \cdots & \bar{r}_{N-2} \\
    \vdots & \ddots    & \ddots & \vdots        \\
    r_{N-1}& \cdots    & r_1    & r_0
    \end{bmatrix},
    \label{eq:AR1}
 \end{aligned}
 \end{equation}
 where
 \begin{equation}
r_l = \operatorname{E}[x_i \bar{x}_{i+l}],\quad  0\leq l \leq N-1, 1\leq i \leq N - l - 1.
\end{equation}
Here, $r_l$ is the $l$-th correlation coefficient of data $\bm{x}$ and $\bar{r}_l$ denotes the conjugate of $r_l$. 

According to the ergodicity of stationary Gaussian process, the correlation coefficients $r_l$ can be approximated by the mean of sample data instead of its statistical expectation, as
\begin{equation}
\widetilde{r}_l = \frac{1}{N} \sum_{i=0}^{N-1- l } {x_i\bar{x}_{i+l}},\quad  0 \leq l \leq N-1.
\end{equation}
By using the correlation coefficient $\bm{\widetilde{r}} = [\widetilde{r}_0, \widetilde{r}_1,\ldots, \widetilde{r}_{N-1}]^{\operatorname{T}}$, the diagonal loading feature matrix is given by
\begin{equation}
\bm{\widetilde{R}} = \bm{\widetilde{r}}\bm{\widetilde{r}}^{\operatorname{H}} + \operatorname{tr}(\bm{\widetilde{r}}\bm{\widetilde{r}}^{\operatorname{H}})\bm{I},
\label{eq:AR2}
\end{equation}

Note that both matrices  \eqref{eq:AR1} and  \eqref{eq:AR2} are HPD, but the feature matrix \eqref{eq:AR1}  is Toeplitz while the one defined by \eqref{eq:AR2} is obtained using diagonal loading. Their difference will be further analyzed for detection problems in the simulation part.

Features of each sample data can be captured  by an HPD matrix by means of \eqref{eq:AR1} or \eqref{eq:AR2}. It is known that the set of HPD matrices forms a differentiable  manifold. Each HPD matrix constructed by only the clutter or the clutter plus target signal corresponds to  a point on this differentiable manifold. Then, the problem of signal detection can be treated as discriminating matrices in the cell under test (CUT) and the CCM on the HPD matrix manifold. In general, the CCM is estimated by a set of secondary data using the GLRT criterion. Given a set of secondary data $\{ \bm{x}_1, \bm{x}_2, \ldots, \bm{x}_K \}$, the SCM estimator is given by
\begin{equation}
\bm{R}_{SCM} = \frac{1}{K}\sum_{k=1}^K \bm{x}_k\bm{x}_k^{\operatorname{H}}, \quad \bm{x}_k \in \mathbb{C}^{N}.
\end{equation}

It is noted that $\bm{R}_{SCM}$ is the arithmetic mean of $K$ autocorrelation matrices $\{\bm{x}_k\bm{x}_k^{\operatorname{H}}\}_{k=1}^K$ with rank one. The SCM estimator $\bm{R}_{SCM}$ is not nonsingular unless $K\geq N$. The performance of SCM estimator often suffers from a severe degradation when the secondary data contains an outlier. Based on these  observations, taken the geometry of HPD matrix manifold into consideration, we replace the arithmetic mean with the geometric mean, and the CCM can be estimated as
\begin{equation}
\bm{R}_\mathcal{G} = \mathcal{G}(\bm{R}_1,  \bm{R}_2, \ldots, \bm{R}_K),
\end{equation}
where $\bm{R}_k$ are given by \eqref{eq:AR1} or \eqref{eq:AR2} with the sample data $\bm{x}_k$, and $\mathcal{G}(\bm{R}_1, \bm{R}_2, \ldots, \bm{R}_K)$ denotes the geometric mean of HPD matrices $\{ \bm{R}_1, \bm{R}_2, \ldots, \bm{R}_K \}$. Then we can realise  signal detection by judging whether the observation is a CCM. The signal detection can be formulated on the HPD matrix manifold as the following hypothesis testing:
\begin{equation}
\left\{
\begin{aligned}
\mathcal{H}_0: \bm{R}=\bm{R}_\mathcal{G}, \\
\mathcal{H}_1: \bm{R}\neq\bm{R}_\mathcal{G}.
\end{aligned}
\right.
\end{equation}

Let us consider the null hypothesis $\mathcal{H}_0: \bm{R}=\bm{R}_\mathcal{G}$ versus the alternative hypothesis $\mathcal{H}_1: \bm{R}\neq\bm{R}_\mathcal{G}$ based on a set of observations $\{\bm{R}_1, \bm{R}_2, \ldots, \bm{R}_K \}$. The CCM is estimated by the geometric mean $\bm{R}_\mathcal{G}$. Then, the problem of signal detection can be understood as to determine the inner of isosurfaces of the HPD matrix manifold determined by a distance or divergence; examples of isosurfaces are available in Fig. \ref{fig:isotropy}.
 The hypothesis  $\mathcal{H}_0$ is rejected if the observation $\bm{R}_D$ of CUT does not belong to the inner of an isosurface with radius $\gamma$. As a consequence, the signal detection can be formulated by,
\begin{equation}
d(\bm{R}_\mathcal{G},\bm{R}_D) \mathop{\gtrless}\limits_{\mathcal{H}_0}^{\mathcal{H}_1} \gamma
\end{equation}
where $d(\bm{R}_\mathcal{G},\bm{R}_D)$ is the dissimilarity between $\bm{R}_\mathcal{G}$ and $\bm{R}_D$ measured by a geometric metric and $\gamma$ denotes the detection threshold that is also the radius of the isosurface. The scheme of signal detection is sketched in Fig. \ref{fig:Detection_scheme}.

It is clear that $d(\bm{R}_\mathcal{G},\bm{R}_D)$ is the detection statistic, denoting  the geometric distance between $\bm{R}_\mathcal{G}$ and $\bm{R}_D$ on the HPD matrix manifold $\mathcal{M}$. Then, we know that the detection performance is closely related to the geometric measure utilized in the detector. Note that the HPD matrix manifold endowed with different divergences will yield different geometric properties, that may  result in different detection performances. Besides, the performance is affected by the geometric mean that is used as the CCM estimate, since different geometric means have different robustness about outliers. In the next context, we will be focused on the definitions and analysis of divergences and geometric means, which determine the performance of signal detection. Specifically, we define the TBD  on the HPD matrix manifold and derive several important TBD means in closed-form.

\begin{figure}[H]
  \centering
  \includegraphics[width=8.3cm]{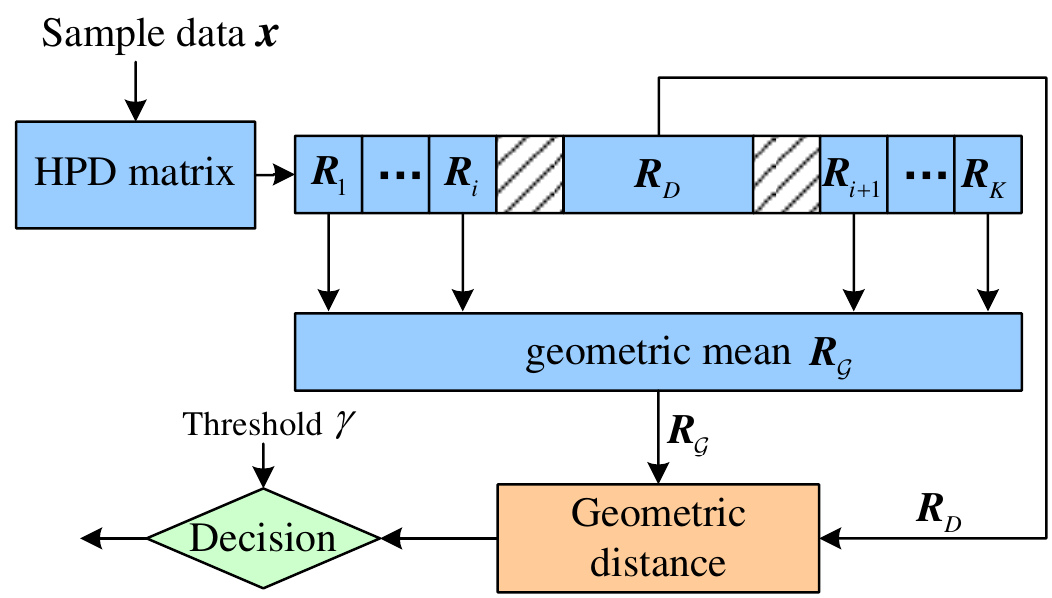}
  \caption{The scheme of signal detection}
  \label{fig:Detection_scheme}
\end{figure}

\section{Information Geometry and TBDs}

The set of all HPD matrices forms a Riemannian manifold equipped with a proper metric, which will allow us to construct efficient algorithms for detection problems. In this section, we will briefly review the theory of classical information geometry and define the TBD divergence on matrix manifolds.

\subsection{Classical information geometry}

The theory of information geometry was firstly established for studying statistical models which are, for instance, associated to a continuous distribution with probability density function $p(x;\bm{\theta})$, where $x$ is the random variable and $\bm{\theta}\in\mathbb{R}^n$ (or a subset of $\mathbb{R}^n$) plays the role of parameters. Assuming the function $p(x;\bm{\theta})$ satisfies the regularity conditions  \cite{amari2000methods,sun2016elementary}, a statistical model is defined as the set of all probability density functions, i.e.
\begin{equation}
\mathcal{M}:=\left\{ p(x;\bm{\theta}) \mid \bm{\theta} \in\mathbb{R}^n, \quad \int p(x;\bm{\theta}) \operatorname{d}\!x=1\right\}.
\end{equation}

A divergence $\operatorname{D}:\mathcal{M}\times \mathcal{M}\rightarrow\mathbb{R}$ is a function measuring the difference/dissimilarity of two elements of $\mathcal{M}$, defined subject to the following properties:
\begin{enumerate}
\item $\operatorname{D}(p_1,p_2)\geq 0$ for all $p_1=p(x;\bm{\theta}_1)$ and $p_2=p(x;\bm{\theta}_2)$.
\item $\operatorname{D}(p_1,p_2)= 0$ if and only if $p_1=p_2$.
\end{enumerate}
The most important classes of divergences are the $f$-divergences and Bregman divergences. For two  elements $p_1=p(x;\bm{\theta}_1)$ and $p_2=p(x;\bm{\theta}_2)$ of $\mathcal{M}$, $f$-divergences are generated through a function $f(z)$, which is convex on $z>0$ and such that $f(1)=0$. An $f$-divergence is defined by the expectation
\begin{equation*}
\begin{aligned}
\operatorname{D}_f(p_1,p_2):&=\operatorname{E}_{p_1}\left[f\left(\frac{p_2}{p_1}\right)\right]=\int p_1 f\left(\frac{p_2}{p_1}\right)\operatorname{d}\!x.
\end{aligned}
\end{equation*}
One well-known example of $f$-divergences is the Kullback--Leibler divergence $\operatorname{D}_{\operatorname{KL}}$  with the function $f$ defined by \cite{kullback1951information}
\begin{equation}
f(z)=-\ln z,\quad z>0.
\end{equation}
It is also known as the relative entropy since
\begin{equation}
\operatorname{D}_{\operatorname{KL}}(p_1,p_2)=H(p_1,p_2)-H(p_1),
\end{equation}
where $H(p_1,p_2)=-\operatorname{E}_{p_1}[\ln p_2]$ is the cross entropy of $p_1$ and $p_2$ and $H(p_1)=-\operatorname{E}_{p_1}[\ln p_1]$ is the entropy of $p_1$. The concept of information entropy was firstly introduced by Shannon in \cite{shannon2001mathematical,shannon1948}.

The Fisher information matrix, firstly introduced by  Fisher in \cite{fisher1922mathematical}, plays the role of a Riemannian metric of the statistical model $\mathcal{M}$. Its components, a symmetric and positive-definite $n\times n$ matrix $g(\bm{\theta})=(g_{ij}(\bm{\theta}))$  can be derived from the infinitesimal behavior of the $f$-divergences; in particular, for the Kullback--Leibler divergence, we have
\begin{equation}
\begin{aligned}
\operatorname{D}_{\operatorname{KL}}&(p(x;\bm{\theta}),p(x;\bm{\theta}+\operatorname{d}\!\bm{\theta}))\\
&=-\int p(x;\bm{\theta}) \ln\left(\frac{p(x;\bm{\theta}+\operatorname{d}\!\bm{\theta})}{p(x;\bm{\theta})}\right)\operatorname{d}\!x\\
&=\frac{1}{2}g_{ij}(\bm{\theta})\operatorname{d}\!\bm{\theta}^i\operatorname{d}\!\bm{\theta}^j+O(\norm{ \operatorname{d}\!\bm{\theta}}^3),
\end{aligned}
\end{equation}
where
\begin{equation}
g_{ij}(\bm{\theta})=\operatorname{E}_{p}\left[\frac{\partial}{\partial \theta^i}\ln p(x;\bm{\theta})~~\frac{\partial}{\partial \theta^j}\ln p(x;\bm{\theta})\right].
\end{equation}
Equipped with the Fisher information matrix or metric $g$, $(\mathcal{M},g)$ becomes a Riemannian manifold, whose Levi-Civita connection $\nabla$ is uniquely given by the torsion-free condition and the following compatibility condition
\begin{equation}
Zg\left(\bm{X},\bm{Y}\right)=g\left(\nabla_{\bm{Z}} \bm{X},\bm{Y}\right)+g\left(\bm{X},\nabla_{\bm{Z}}\bm{Y}\right).
\end{equation}
Here $\bm{X},\bm{Y},\bm{Z}$ are vector fields on $(\mathcal{M},g)$. Chentsov \cite{Chentsov1972}, Efron \cite{Efron1975Defining},   Amari \cite{97844712097,amari2000methods}, etc. extended the above compatibility condition to the existence of a one-parameter family of affine connections $\nabla^{(\alpha)}$ ($\alpha\in\mathbb{R}$) satisfying the duality condition
\begin{equation}
\bm{Z}g\left(\bm{X},\bm{Y}\right)=g\left(\nabla^{(\alpha)}_{\bm{Z}} \bm{X},\bm{Y}\right)+g\left(\bm{X},\nabla^{(-\alpha)}_{\bm{Z}}\bm{Y}\right).
\end{equation}
Here $\nabla^{(\alpha)}$ and $\nabla^{(-\alpha)}$ is a pair of dual connections. When $\alpha=0$, it reduces to the Levi-Civita connection.

\subsection{TBDs on  matrix manifolds}

Let $F$ be a differentiable and strictly convex function defined on a convex domain of $\mathbb{R}^n$. A Bregman divergence, introduced by Bregman \cite{bregman1967relaxation},  measures the difference between the value of $F$ at a point $\bm{x}\in\mathbb{R}^n$ and the linear approximation of $F$ around point $\bm{y}$ evaluated at the point $\bm{x}$, namely
\begin{equation}
\operatorname{B}_F(\bm{x},\bm{y}):=F(\bm{x})-F(\bm{y})-\langle \nabla F(\bm{y}), \bm{x}-\bm{y}\rangle,
\end{equation}
where $\nabla F$ denotes the gradient of $F$ and $\langle\cdot,\cdot\rangle$ is the natural inner product of two vectors. In fact, the Kullback--Leibler divergence is also a special case of the Bregman divergence \cite{zhang2004divergence}. If the space of $\bm{x}:=\bm{\theta}$ is the parameter space of a model or manifold, the Bregman divergence induces a Riemannian metric $\left(g_{ij}(\bm{\theta})\right)=\operatorname{Hess}F(\bm{\theta})$ and a family of information-geometric one-parameter dual connections. The function $F$ is hence sometimes called a potential function.

In recent years, TBDs were introduced and found to be more efficient in dealing with practical problems such as shape retrieval and diffusion tensor image \cite{vemuri2010total,liu2010total}. The TBD $\delta_F$ is defined between two points $\bm{x},\bm{y}$ on a convex domain of $\mathbb{R}^n$ for a differentiable and strictly convex function $F$, as follows
\begin{equation}
\delta_F(\bm{x},\bm{y}):=\frac{F(\bm{x})-F(\bm{y})-\langle \nabla F(\bm{y}), \bm{x}-\bm{y}\rangle}{\sqrt{1+\norm{ \nabla F(\bm{y})}^2}}.
\end{equation}
As a scaling of the Bregman divergence, it shares a lot of similarities as the Bregman divergence, such as convexity (about the first argument).
One aim of this paper is to extend the definition of TBD to matrix manifolds and in particular to the HPD matrix manifold.

The theory of matrix groups is essential in applied sciences. Matrix information geometry  studies the Riemannian-geometric structures of matrix groups, which have been found fundamentally crucial in linear and nonlinear problems.  For the general linear group $GL(N,\mathbb{F})$ of $N\times N$ invertible matrices,  where $\mathbb{F}$ is either $\mathbb{R}$ or $\mathbb{C}$, one can define the following  metric or inner product 
\begin{equation}\label{eq:glmetric}
\langle \bm{X},\bm{Y}\rangle:=\operatorname{tr}(\bm{X}^{\operatorname{H}}\bm{Y}),
\end{equation}
where $\bm{X},\bm{Y}\in GL(N,\mathbb{F})$ and $\bm{X}^{\operatorname{H}}$ denotes the conjugate transpose of $\bm{X}$ in the complex case or simply the transpose of $\bm{X}$ in the real case. This metric is called the Frobenius inner product or the Hilbert--Schmidt inner product. An induced Riemannian metric can be defined at its tangent bundle leading to a unique Levi-Civita connection; this gives its Riemannian structure.  An AIRM of the HPD matrix manifold will be given later in this section.  In many cases, such as subgroups of $GL(N,\mathbb{F})$ with better geometric or topological properties, one may define various metrics and even dual connections as have been greatly investigated for statistical models.

The Bregman divergence for matrices $\bm{X}$ and $\bm{Y}$ is defined as (e.g., \cite{dhillon2008matrix})
\begin{equation}\label{eq:BDmatrix}
\operatorname{B}_F(\bm{X},\bm{Y}):=F(\bm{X})-F(\bm{Y})-\langle \nabla F(\bm{Y}),\bm{X}-\bm{Y}\rangle,
\end{equation}
where $\langle\cdot,\cdot\rangle$ is the Frobenius inner product \eqref{eq:glmetric} and $F:GL(N,\mathbb{F})\rightarrow \mathbb{F}$ is a strictly convex and differentiable function. It can then be immediately generalized to a TBD as the following definition.
\begin{defn}
The TBD of two matrices $\bm{X},\bm{Y}\in GL(N,\mathbb{F})$ is defined as
\begin{equation}\label{eq:TBDmatrix}
\begin{aligned}
\delta_F(\bm{X},\bm{Y})=\frac{F(\bm{X})-F(\bm{Y})-\langle \nabla F(\bm{Y}),\bm{X}-\bm{Y}\rangle}{\sqrt{1+\norm{\nabla F(\bm{Y})}^2}},
\end{aligned}
\end{equation}
where $\norm{\bm{X}}:=\sqrt{\langle \bm{X}, \bm{X}\rangle}$.
\end{defn}

Next, let us consider several examples.

\begin{prop}
Let $F(\bm{X})=\frac{1}{2}\norm{\bm{X}}^2$. The corresponding TBD, called the total square loss (TSL), is given by
\begin{equation}\label{eq:TBD_TSL}
\begin{aligned}
\delta_F(\bm{X},\bm{Y})=\frac{1}{2}\frac{\norm{ \bm{X}-\bm{Y}}^2}{\sqrt{1+\norm{ \bm{Y}}^2}}.
\end{aligned}
\end{equation}
\end{prop}

\begin{proof}

The gradient of $F(\bm{X})$ associated to the Frobenius inner product is given by the Fr\'echet derivative
\begin{equation*}
\begin{aligned}
\langle \nabla F(\bm{X}),\bm{Y}\rangle:&=\frac{\operatorname{d}}{\operatorname{d}\!\varepsilon}\Big|_{\varepsilon=0} F(\bm{X}+\varepsilon \bm{Y})\\
&=\frac{1}{2}\frac{\operatorname{d}}{\operatorname{d}\!\varepsilon}\Big|_{\varepsilon=0}\operatorname{tr}\left((\bm{X}+\varepsilon \bm{Y})^{\operatorname{H}}(\bm{X}+\varepsilon \bm{Y})\right)\\
&=\frac{1}{2}\operatorname{tr}(\bm{X}^{\operatorname{H}}\bm{Y})+\frac{1}{2}\operatorname{tr}(\bm{Y}^{\operatorname{H}}\bm{X})\\
&=\langle \bm{X},\bm{Y}\rangle.
\end{aligned}
\end{equation*}
Namely, $\nabla F(\bm{X})=\bm{X}$.
Then the Bregman divergence \eqref{eq:BDmatrix} is given by
\begin{equation}
\operatorname{B}_F(\bm{X},\bm{Y})=\frac{1}{2}\norm{ \bm{X}-\bm{Y}}^2.
\end{equation}
Consequently, the TSL is obtained.
\end{proof}

\begin{prop}
Let $F(\bm{X})=-\ln\det \bm{X}$, which is induced from the function  $- \ln x$; see \cite{dhillon2008matrix}. The total log-determinant (TLD) divergence is given by
\begin{equation}\label{eq:TBD_TLD}
\delta_F(\bm{X},\bm{Y})=\frac{\ln \det(\bm{Y}\bm{X}^{-1})+\operatorname{tr}(\bm{Y}^{-1}\bm{X})-N}{\sqrt{1+\norm{ \bm{Y}^{-\operatorname{H}}}^2}}.
\end{equation}
Note that we assumed $\bm{X}$ and $\bm{Y}$ satisfy necessary conditions to avoid singularity.
\end{prop}

\begin{proof}
Similar computation gives the gradient by
\begin{equation}
\langle \nabla F(\bm{X}),\bm{Y}\rangle=-\operatorname{tr}(\bm{X}^{-1}\bm{Y}).
\end{equation}
Namely,
\begin{equation}
\nabla F(\bm{X})=-\bm{X}^{-\operatorname{H}},
\end{equation}
and the Bregman divergence is given by
\begin{equation}
\begin{aligned}
\operatorname{B}_F(\bm{X},\bm{Y})
=\ln \det(\bm{Y}\bm{X}^{-1})+\operatorname{tr}(\bm{Y}^{-1}\bm{X})-N.
\end{aligned}
\end{equation}
Here $N$ is the dimension of matrices $\bm{X},\bm{Y}$. Then the TLD divergence can be derived.

\end{proof}

It is known that if $\bm{X}$ is an invertible matrix that does not have eigenvalues in the closed negative real line, then there exists a unique logarithm with eigenvalues lying in the strip $\{z\in \mathbb{C}\mid -\pi<\operatorname{Im}(z)<\pi\}$ \cite{Hig2008}. This logarithm is called the principal logarithm and denoted by $\operatorname{Log}\bm{X}$.

\begin{prop}\label{prop:tvn}
Suppose $\bm{X}$ is invertible and has no eigenvalues lying in the negative real line and define $$F(\bm{X})=\operatorname{tr}(\bm{X}\operatorname{Log}\bm{X}-\bm{X}),$$ which is induced from the function $x\ln x-x$ (e.g., \cite{dhillon2008matrix}). Then, the total von Neumann (TVN) divergence is given by
\begin{equation}\label{eq:TBD_TVN}
\delta_F(\bm{X},\bm{Y})=\frac{\operatorname{tr}\left( \bm{X}(\operatorname{Log} \bm{X} - \operatorname{Log} \bm{Y} ) - \bm{X} + \bm{Y} \right)}{\sqrt{1+\norm{ (\operatorname{Log} \bm{Y})^{\operatorname{H}}}^2}}.
\end{equation}
\end{prop}

 \begin{proof}
See Appendix \ref{appen:A}.
\end{proof}

For the HPD matrix manifold
\begin{equation*}
\begin{aligned}
\mathscr{P}(N,\mathbb{C}):=\left\{  \bm{X}\in GL(N,\mathbb{C}) \mid \bm{z}^{\operatorname{H}}\bm{X}\bm{z}>0, \forall \bm{z}\in\mathbb{C}^N/ \{0\}\right\},
\end{aligned}
\end{equation*}
an AIRM at a point $\bm{P}\in\mathscr{P}(N,\mathbb{C})$ is defined by
\begin{equation}\label{eq:HPDmetric}
g_{\bm{P}}(\bm{A},\bm{B}):=\operatorname{tr}\left(\bm{P}^{-1}\bm{A}\bm{P}^{-1}\bm{B}\right),
\end{equation}
where $\bm{A},\bm{B}\in T_{\bm{P}}\mathscr{P}(N,\mathbb{C})$ and hence $\bm{P}^{-1}\bm{A}, \bm{P}^{-1}\bm{B}\in T_{\bm{I}}\mathscr{P}(N,\mathbb{C})$. The induced distance (called geodesic distance or Riemannian distance) between $\bm{X}$ and $\bm{Y}$ is given by
\begin{equation}
\begin{aligned}
d(\bm{X},\bm{Y}) = \norm{\operatorname{Log}\left(\bm{X}^{-\frac12}\bm{Y}\bm{X}^{-\frac12}\right)}.
\end{aligned}
\end{equation}
Note that the AIRM is consistent with the Frobenius inner product \eqref{eq:glmetric} when restricted to the identity $\bm{I}$ of $\mathscr{P}(N,\mathbb{C})$.  Its Lie algebra $\mathfrak{P}(N,\mathbb{C})=T_{\bm{I}} \mathscr{P}(N,\mathbb{C})$ consists of all Hermitian matrices
\begin{equation}
\mathfrak{P}(N,\mathbb{C})=\left\{\bm{X}\in GL(N,\mathbb{C})\mid \bm{X}^{\operatorname{H}}=\bm{X} \right\}.
\end{equation}
From now on, we will only focus on the HPD matrix manifold $\mathscr{P}(N,\mathbb{C})$. The Bregman divergence \eqref{eq:BDmatrix} and TBD \eqref{eq:TBDmatrix} defined on $GL(N,\mathbb{C})$ can be directly restricted to HPD matrices.

\begin{figure}[H]
  \centering
  \includegraphics[width=8.3cm]{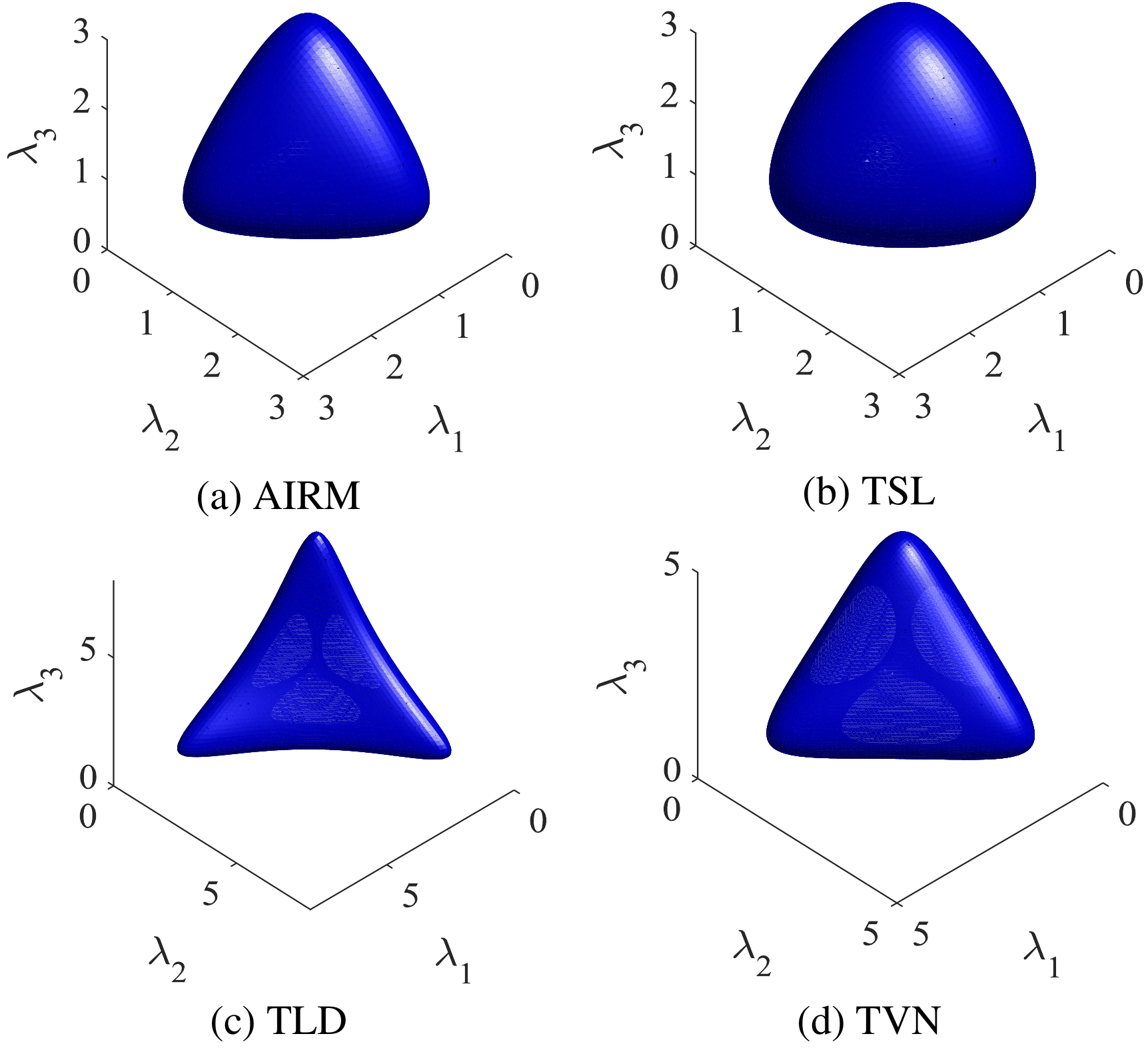}
  \caption{Isosurfaces in $\mathscr{P}(3,\mathbb{C})$ centered at $\bm{I}$ with unit radius, e.g., $\{\bm{X}\in \mathscr{P}(3,\mathbb{C})\mid \delta_F(\bm{I},\bm{X})=1 \}$ where $\lambda_1$, $\lambda_2$ and $\lambda_3$ are eigenvalues of $\bm{X}$}
  \label{fig:isotropy}
\end{figure}

To analyze  the difference of these divergences defined on a matrix manifold, we show the plots of three-dimensional isosurfaces associated with the AIRM and TBD centered at the identity. As shown in Fig. \ref{fig:isotropy}, all the TBDs and AIRM induce non-spherical convex balls. The shapes of isosurfaces  reflecting  geometric properties  of a matrix manifold are totally different.

\section{TBDs and robustness analysis}

In this section, we first recall the definition of arithmetic mean for a set of real numbers and define the geometric mean associated with TBD on the HPD matrix manifold in a similar way. Specifically, three TBD means, associated with the TSL, TLD and TVN divergences, are derived by considering the stationary condition of relevant optimization problems. Finally, influence functions are defined to analyze   robustness of the TBD means.

\subsection{TBD means for a set of HPD matrices}

For a set of $m$ real numbers $\{ x_1, x_2,\ldots, x_m\}$, the well-known arithmetic mean is  
\begin{equation}
\bar{x} := \frac{1}{m} \sum_{i=1}^{m} x_i.
\end{equation}
From a geometric viewpoint, the arithmetic mean can be obtained by considering the minimum of the sum of squared distances, namely
\begin{equation}
\bar{x} := \underset{x \in \mathbb{R}}{\operatorname{argmin}} \sum_{i=1}^{m} | x - x_i |^2,
\end{equation}
where $| x - x_i |$ is the absolute value of $x - x_i$, denoting the distance of two points on the real line.

This can be generalized to define   TBD means for a set of HPD matrices.

\begin{defn}
Let $F$ be a strictly convex and differentiable function, and let $\delta_F$ be the corresponding TBD. The TBD mean for a set of $m$ number of HPD matrices $\{ \bm{X}_1, \bm{X}_2, \ldots, \bm{X}_m\}$ is defined as follows:
\begin{equation}\label{eq:TBDMeanDef}
\overline{\bm{X}} := \underset{\bm{X}\in \mathscr{P}(n,\mathbb{C})}{\operatorname{argmin}} \frac{1}{m}\sum_{i=1}^{m} \delta_F(\bm{X}, \bm{X}_i).
\end{equation}
\end{defn}

Since $F(\bm{X})$ is strictly convex,  the function $\frac{1}{m}\sum\limits_{i=1}^m\delta_F(\bm{X},\bm{X}_i)$ is also strictly convex about $\bm{X}$. Therefore, if the TBD mean \eqref{eq:TBDMeanDef} exists, then it is unique. It can be calculated using the stationary condition as follows. Note that to assure its existence, the function $F$ (and hence the TBD) should be defined in a compact matrix space containing the HPD matrix manifold. Also be noted that we add the coefficient $\frac1m$ for late convenience; this obviously will not affect  value of the mean.


\begin{prop}
If the TBD mean \eqref{eq:TBDMeanDef} exists, then it solves the algebraic equation
\begin{equation}\label{eq:TBDMeanSolution}
\nabla F(\bm{X}) = {\sum\limits_{i=1}^{m} \frac{\nabla F(\bm{X}_i)}{\sqrt{1+\norm{\nabla F(\bm{X}_i)}^2}}}\Big{/}{\sum\limits_{j=1}^{m} \frac{1}{\sqrt{1+\norm{\nabla F(\bm{X}_j)}^2}}}.
\end{equation}
\end{prop}

\begin{proof}

Denote $G(\bm{X})$ as the function to be minimized, namely
\begin{equation*}
G(\bm{X}) :=\frac1m \sum_{i=1}^{m} \frac{F(\bm{X})-F(\bm{X}_i)-\langle \nabla F(\bm{X}_i),\bm{X}-\bm{X}_i\rangle}{\sqrt{1+\norm{\nabla F(\bm{X}_i)}^2}};
\end{equation*}
its gradient can be immediately calculated and we have
\begin{equation}
\nabla G(\bm{X}) = \frac1m \sum_{i=1}^{m} \frac{\nabla F(\bm{X}) - \nabla F(\bm{X}_i)}{\sqrt{1+\norm{\nabla F(\bm{X}_i)}^2}}.
\end{equation}
Letting $\nabla G(\bm{X}) = \bm{0}$ completes the proof.

%
%

\end{proof}

Next, we will show the TBD means with respect to the TSL, TLD and TVN divergences, respectively.

\begin{prop} \label{prop:tslmean}
The TSL mean for a set of $m$ HPD matrices $\{ \bm{X}_1, \bm{X}_2, \ldots, \bm{X}_m \}$ is given by
\begin{equation*}
\overline{\bm{X}}_{TSL} = \sum_{j=1}^m\sqrt{1+\norm{\bm{X}_j}^2}\times \sum_{i=1}^{m} \frac{\bm{X}_i}{\sqrt{1+\norm{ \bm{X}_i }^2}}.
\end{equation*}
\end{prop}

\begin{prop}\label{prop:tldmean}
The TLD mean for a set of $m$ HPD matrices $\{ \bm{X}_1, \bm{X}_2, \ldots, \bm{X}_m \}$ is given by
\begin{equation*}
\overline{\bm{X}}_{TLD} = \left(\sum_{j=1}^m\sqrt{1+\norm{\bm{X}_j^{-1}}^2} \times \sum_{i=1}^{m}\frac{\bm{X}_i^{-1}}{\sqrt{1+\norm{\bm{X}_i^{-1}}^2} }\right)^{-1}.
\end{equation*}
\end{prop}

\begin{prop}\label{prop:tvnmean}
The TVN mean for a set of $m$ HPD matrices $\{ \bm{X}_1, \bm{X}_2, \ldots, \bm{X}_m \}$ is given by
\begin{equation*}
\begin{aligned}
\overline{\bm{X}}_{TVN} = \exp\left(\sum_{j=1}^m\mu_j\times \sum_{i=1}^{m} \frac{ \operatorname{Log}\bm{X}_i}{\mu_i}\right),
\end{aligned}
\end{equation*}
where
\begin{equation*}
 \mu_i =\sqrt{1+\norm{ \operatorname{Log}\bm{X}_i }^2},\quad i=1,2,\ldots,m.
\end{equation*}
\end{prop}

Proofs of the Propositions \ref{prop:tslmean}, \ref{prop:tldmean} and \ref{prop:tvnmean} are straightforward and  are omitted here.


\subsection{Influence functions}

We define an influence function that describes the effect of outliers on the estimate accuracy of the TBD mean. The influence function can be used for analysing the robustness of TBD means when the HPD data are contaminated by outliers. Let $\bm{\overline{X}}$ be the TBD mean (or AIRM mean) of $m$ HPD matrices $\{ \bm{X}_1, \bm{X}_2, \ldots, \bm{X}_m \}$, and $\bm{\widehat{X}}$ is the TBD mean  (or AIRM mean) of the contaminated HPD data obtained by adding $n$ outliers, that are also HPD matrices $\{ \bm{P}_1, \bm{P}_2, \ldots, \bm{P}_n \}$, with a weight $\varepsilon (\varepsilon \ll 1)$ into these $m$ HPD matrices. Then, $\bm{\widehat{X}}$ can be defined as a perturbation
$$\bm{\widehat{X}} = \bm{\overline{X}} + \varepsilon \bm{H}(\overline{\bm{X}},\bm{P}_1,\bm{P}_2,\ldots,\bm{P}_n)+O(\varepsilon^2),$$
 and the norm
 $$h(\overline{\bm{X}},\bm{P}_1,\bm{P}_2,\ldots,\bm{P}_n):=
\norm{ \bm{H}(\overline{\bm{X}},\bm{P}_1,\bm{P}_2,\ldots,\bm{P}_n)} $$ is defined as the influence function. 
The influence functions of the AIRM and TBD means are given as follows.

\begin{prop}\label{prop:TBDae}
The influence function of the AIRM mean for $m$ HPD matrices $\{ \bm{X}_1, \bm{X}_2, \ldots, \bm{X}_m \}$ and $n$ outliers $\{\bm{P}_1, \bm{P}_2, \ldots, \bm{P}_n\}$ with a weight $\varepsilon (\varepsilon \ll 1)$, is $h=\norm{\bm{H}}$ where
\begin{equation}\label{eq:AIRMInfluFun}
\bm{H}=- \frac{1}{n}\sum_{j=1}^{n}\frac{\bm{\overline{X}}  \operatorname{Log}(\bm{P}_j^{-1}\bm{\overline{X}})+\operatorname{Log}(\bm{\overline{X}} \bm{P}_j^{-1})\bm{\overline{X}}  }{2}.
\end{equation}
\end{prop}

 \begin{proof}
See Appendix \ref{appen:B}.
\end{proof}

Now let  $F$ be a differentiable and strictly convex function and let us consider the corresponding  TBD mean  $\bm{\overline{X}}$ of $m$ HPD matrices $\{ \bm{X}_1, \bm{X}_2, \ldots, \bm{X}_m \}$, and TBD mean $\bm{\widehat{X}}$ of the contaminated HPD data obtained by adding $n$ outliers, i.e., HPD matrices $\{ \bm{P}_1, \bm{P}_2, \ldots, \bm{P}_n \}$.

Define $G(\bm{X})$ as the objective function to be minimized for the enlarged set of data, namely
\begin{equation*}
\begin{aligned}
G(\bm{X}) := (1-\varepsilon)\frac{1}{m}\sum_{i=1}^{m}\delta_F(\bm{X},\bm{X}_i) + \varepsilon\frac{1}{n}\sum_{j=1}^{n}\delta_F(\bm{X},\bm{P}_j).
\end{aligned}
\end{equation*}
Since $\bm{\widehat{X}}$ is the TBD mean of $m$ HPD matrices $\{ \bm{X}_1, \bm{X}_2, \ldots, \bm{X}_m \}$ and $n$ outliers $\{\bm{P}_1, \bm{P}_2, \ldots, \bm{P}_n\}$,  we have $\nabla G(\bm{\widehat{X}}) = \bm{0}$, i.e.,
\begin{equation}\label{eq:TBD1}
\begin{aligned}
&(1-\varepsilon)\frac{1}{m}\sum_{i=1}^{m} \frac{\nabla F(\bm{\widehat{X}})- \nabla F(\bm{X}_i)}{\sqrt{1+\norm{\nabla F(\bm{X}_i)}^2}} \\
&~~~~~~~~~~~~+ \varepsilon \frac{1}{n}\sum_{j=1}^{n} \frac{\nabla F(\bm{\widehat{X}})- \nabla F(\bm{P}_j)}{\sqrt{1+\norm{\nabla F(\bm{P}_j)}^2}} = \bm{0}.
\end{aligned}
\end{equation}
Similarly, we differentiate $\nabla G(\widehat{\bm{X}})=\bm{0}$ with respect to $\varepsilon$ at $\varepsilon=0$, leading to
\begin{equation}\label{eq:lp111}
\begin{aligned}
\frac{1}{m}\sum_{i=1}^m&\frac{1}{\sqrt{1+\norm{\nabla F(\bm{X}_i)}^2}}\frac{\operatorname{d}}{\operatorname{d}\!\varepsilon}\Big|_{\varepsilon=0}\nabla F(\widehat{\bm{X}})\\
&~~~~~~+\frac{1}{n}\sum_{j=1}^n\frac{\nabla F(\overline{\bm{X}})-\nabla F(\bm{P}_j)}{\sqrt{1+\norm{\nabla F(\bm{P}_j)}^2}}=\bm{0}.
\end{aligned}
\end{equation}
Taking $\widehat{\bm{X}}=\overline{\bm{X}}+\varepsilon \bm{H}+O(\varepsilon^2)$ into account will give us the influence function. We will now study the TSL, TLD and TVN means, respectively.

\begin{prop}
For $F(\bm{X}) = \frac{1}{2}\norm{ \bm{X} }^2$ corresponding to the TSL, we have $\nabla F(\bm{X}) = \bm{X}$ and hence
 $$\frac{\operatorname{d}}{\operatorname{d}\!\varepsilon}\Big|_{\varepsilon=0}\nabla F(\widehat{\bm{X}})=\bm{H}(\overline{\bm{X}},\bm{P}_1,\bm{P}_2,\ldots,\bm{P}_n).$$
Therefore, the influence function of the TSL mean  is given by $h=\norm{\bm{H}}$ where
\begin{equation}\label{eq:influ_of_TSL}
\bm{H}= -\frac{m}{n}\left(\sum_{i=1}^{m} \frac{1}{\sqrt{1+\norm{ \bm{X}_i }^2}}\right)^{-1} \sum_{j=1}^{n} \frac{\bm{\overline{X}}-\bm{P}_j}{\sqrt{1+\norm{ \bm{P}_j }^2}}.
\end{equation}
\end{prop}

\begin{prop}
For $F(\bm{X})=-\ln \det\bm{X}$ corresponding to the TLD divergence, we have $\nabla F(\bm{X})=-\bm{X}^{-1}$
and
$$\frac{\operatorname{d}}{\operatorname{d}\!\varepsilon}\Big|_{\varepsilon=0}\nabla F(\widehat{\bm{X}})=\overline{\bm{X}}^{-1}\bm{H}\overline{\bm{X}}^{-1}.$$
 The influence function of the TLD mean is given by $h=\norm{\bm{H}}$ where
\begin{equation}\label{eq:influ_of_TLD}
\bm{H} =  \frac{m}{n}\left(\sum_{i=1}^{m}\frac{1}{\sqrt{1+\norm{ \bm{X}_i^{-1} }^2}}\right)^{-1} \sum_{j=1}^{n}\frac{\bm{\overline{X}}\left(\overline{\bm{X}}^{-1}-\bm{P}_j^{-1}\right)\overline{\bm{X}}}{\sqrt{1+\norm{ \bm{P}_j^{-1} }^2}}.
\end{equation}
\end{prop}

Note that we used the fact
\begin{equation*}
\begin{aligned}
\bm{0}&=\frac{\operatorname{d}}{\operatorname{d}\!\varepsilon}\bm{I}=\frac{\operatorname{d}}{\operatorname{d}\!\varepsilon}\left(\widehat{\bm{X}}\widehat{\bm{X}}^{-1}\right)=\frac{\operatorname{d}}{\operatorname{d}\!\varepsilon}\widehat{\bm{X}}\widehat{\bm{X}}^{-1}+\widehat{\bm{X}}\frac{\operatorname{d}}{\operatorname{d}\!\varepsilon}\widehat{\bm{X}}^{-1}
\end{aligned}
\end{equation*}
and hence
\begin{equation*}
\frac{\operatorname{d}}{\operatorname{d}\!\varepsilon}\widehat{\bm{X}}^{-1}=-\widehat{\bm{X}}^{-1}\frac{\operatorname{d}}{\operatorname{d}\!\varepsilon}\widehat{\bm{X}}\widehat{\bm{X}}^{-1}.
\end{equation*}

\begin{prop}
When $F(\bm{X})=\operatorname{tr}(\bm{X}\operatorname{Log}\bm{X}-\bm{X})$, the corresponding divergence is the TVN divergence. We have
$\nabla F(\bm{X})=\operatorname{Log}\bm{X}$
and
$$\frac{\operatorname{d}}{\operatorname{d}\!\varepsilon}\Big|_{\varepsilon=0}\nabla F(\widehat{\bm{X}})
=\int_0^1[(\overline{\bm{X}}-\bm{I})s+\bm{I}]^{-1}\bm{H}[(\overline{\bm{X}}-\bm{I})s+\bm{I}]^{-1}\operatorname{d}\!s.$$
Substituting $F(\bm{X})$ into \eqref{eq:lp111} and taking trace on both sides give us
\begin{equation*}
\begin{aligned}
\operatorname{tr}(\overline{\bm{X}}^{-1}\bm{H})&=-\frac{m}{n}\left(\sum_{i=1}^{m}\frac{1}{\sqrt{1+\norm{ \operatorname{Log}\bm{X}_i }^2}}\right)^{-1}\\
&~~~~\times \sum_{j=1}^{n}\frac{\operatorname{tr}\left(\operatorname{Log}\overline{\bm{X}}-\operatorname{Log}\bm{{P}_j}\right)}{\sqrt{1+\norm{ \operatorname{Log}\bm{P}_j }^2}}.
\end{aligned}
\end{equation*}
Assuming the arbitrarity of  $\overline{\bm{X}}$, we obtain
the influence function of the TVN mean as $h=\norm{\bm{H}}$ in which we choose
\begin{equation}\label{eq:influ_of_TVN}
\begin{aligned}
\bm{H}&=-\frac{m}{n}\left(\sum_{i=1}^{m}\frac{1}{\sqrt{1+\norm{ \operatorname{Log}\bm{X}_i }^2}}\right)^{-1}\\
&~~~~\times \sum_{j=1}^{n}\frac{\overline{\bm{X}}^{1/2}\left(\operatorname{Log}\overline{\bm{X}}-\operatorname{Log}\bm{{P}_j}\right)\overline{\bm{X}}^{1/2}}{\sqrt{1+\norm{ \operatorname{Log}\bm{P}_j }^2}}.
\end{aligned}
\end{equation}
\end{prop}

It is obvious that the influence function of AIRM mean (see Eq. \eqref{eq:AIRMInfluFun}), that is,
\begin{equation}
h(\overline{\bm{X}},\bm{P}_1,\bm{P}_2,\ldots,\bm{P}_n)=\norm{ \frac{1}{n}\sum_{j=1}^{n}\bm{\overline{X}}  \operatorname{Log}(\bm{P}_j^{-1}\bm{\overline{X}})},
 \end{equation}
is not bounded with respect to its arguments $\bm{P}_1,\bm{P}_2,\ldots,\bm{P}_n$. However,  the influence functions of the TSL, TLD and TVN means are all bounded. It suffices to show that the second term of \eqref{eq:lp111} is bounded by noting that the first term of \eqref{eq:lp111} is independent of  $\bm{P}_1,\bm{P}_2,\ldots,\bm{P}_n$. We have
  \begin{equation*}
  \begin{aligned}
  \norm{\frac{1}{n}\sum_{j=1}^n\frac{\nabla F(\overline{\bm{X}})-\nabla F(\bm{P}_j)}{\sqrt{1+\norm{\nabla F(\bm{P}_j)}^2}}}&\leq \frac{1}{n} \sum_{j=1}^n\frac{\norm{\nabla F(\overline{\bm{X}})}+\norm{\nabla F(\bm{P}_j)}}{\sqrt{1+\norm{\nabla F(\bm{P}_j)}^2}}\\
  &\leq \frac{1}{n} \sum_{j=1}^n\left(\norm{\nabla F(\overline{\bm{X}})}+1\right)\\
  &=\norm{\nabla F(\overline{\bm{X}})}+1.
  \end{aligned}
  \end{equation*}
Consequently, the three influence functions are all bounded with respect to $\bm{P}_1,\bm{P}_2,\ldots,\bm{P}_n$.

\subsection{Robustness Analysis}

To validate the advantage of the robustness of TBD mean, numerical simulations are given under different number of sample data and outliers. The results are compared with the AIRM mean and SCM. In the simulation, the sample data is generated according to an $N$-dimensional complex circular Gaussian distribution with zero-mean and a known covariance matrix $\bm{\Sigma}$ given by
\begin{equation}
\bm{\Sigma} = \bm{\Sigma_0} + \bm{I},
\end{equation}
where
\begin{equation*}
 \bm{\Sigma_0}(i,k) = \sigma_c^2\rho^{| i-k|}e^{\operatorname{i}2\pi f_d(i-k)},\quad  i,k = 1,2,\ldots, N.
\end{equation*}
 Here, $\rho$ is the one-lag correlation coefficient,  $\sigma_c$ denotes the clutter-to-noise power ratio, and $f_d$ is the clutter normalized Doppler frequency. In this part, we set $\rho=0.9$, $\sigma_c^2=20 dB$, $f_d=0.2$, and the dimension of covariance matrix $N = 8$.

 We first generate $K$ number of sample data, each of that is an $N$-dimensional vector. The TBD mean and the AIRM mean $\bm{\overline{R}}$ are computed according to $K$ HPD matrices derived by the generated $K$ sample data, respectively. The HPD matrix can be computed via Eq. \eqref{eq:AR1}. Then, the offset error between the real matrix $\bm{\Sigma}$ and the estimated matrix $\bm{\overline{R}}$ is computed, which is defined as $$\norm{ \bm{\overline{R}}- \bm{\Sigma}}\Big{/}\norm{\bm{\Sigma}}.$$ The computation is repeated $1000$ times for each value and the averaged results are shown in Fig. \ref{fig:Robustnes_Of_Num_Of_Sample}.
\begin{figure}[H]
\centering
{\includegraphics[width=8.5cm]{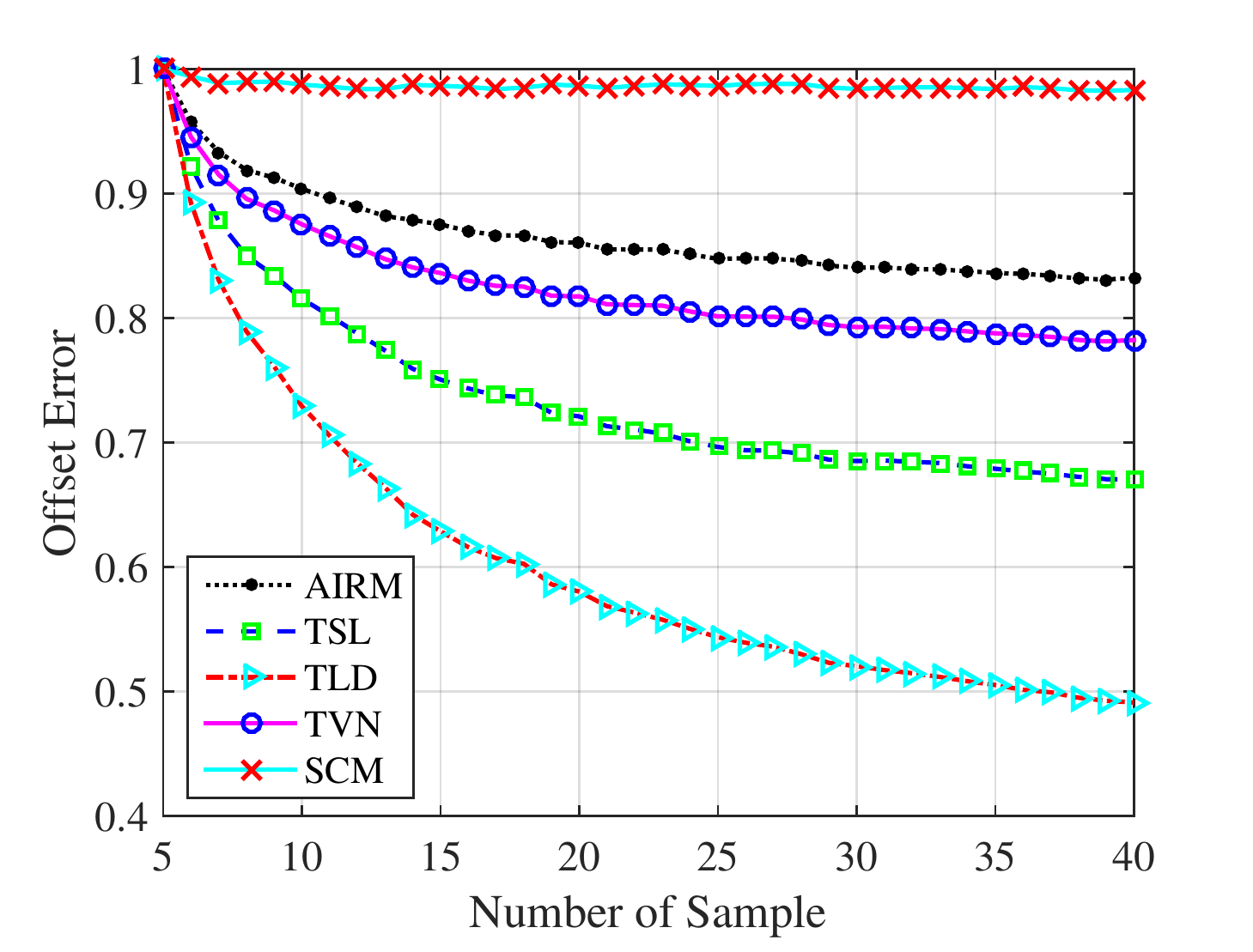}}
\caption{Offset errors of TBD and AIRM means under different number of sample data}
\label{fig:Robustnes_Of_Num_Of_Sample}
\end{figure}

In Fig. \ref{fig:Robustnes_Of_Num_Of_Sample}, offset errors of all the four geometric means decrease with the number of sample data while the offset error of the SCM is the least sensitive. It implies that estimation error decreases fast for geometric means when the number of sample data increases, while for the SCM, this effect is not obvious. It is also observed that the offset errors of geometric means are smaller than that of the SCM. Namely, the geometric means are more robust to the number of sample data than the SCM. Moreover, the TBD means have better robustness than the AIRM mean, and the TLD has the best robustness followed by the TSL mean.

To show  robustness of the TBD mean to outliers, we add $M$ outliers into the generated $50$ sample data. The outlier is modeled as $\bm{x}=\alpha \bm{p}+\bm{c}$, where $\bm{p}$ is the steering vector, and $\alpha$ denotes the amplitude coefficient. The signal/interference to clutter ratio (SCR/ICR) is defined as
\begin{equation}
\textsc{SCR} = |\alpha|^2\bm{p}^{\operatorname{H}}\bm{R}^{-1}\bm{p}.
\end{equation}
Here, the SCR is set to $40$ dB. Values of the influence functions for the TSL, TLD and TVN means, as well as for the AIRM mean can be computed using Eqs. \eqref{eq:influ_of_TSL},  \eqref{eq:influ_of_TLD},  \eqref{eq:influ_of_TVN} and  \eqref{eq:AIRMInfluFun}, respectively. Value of the influence function of the SCM can be calculated in a similar manner. Again, we repeat the computation $1000$ times and take their average. The number of outliers $M$ varies from $1$ to $40$.

Fig. \ref{fig:Robustnes_Of_Num_Of_Outlier} shows the influence values of the TBD means, the AIRM mean and the SCM under different number of outliers. Obviously, the geometric means have lower influence values than the SCM and the influences values of the TBD means are smaller than that of the AIRM mean. It implies that the geometric means are more robust to outliers than the SCM and the TBD means have better robustness than the AIRM mean. In addition, the TLD mean is similar  to the TSL mean about robustness and both of them are more robust to outliers compared with the TVN mean.

\begin{figure}[H]
\centering
{\includegraphics[width=8.2cm]{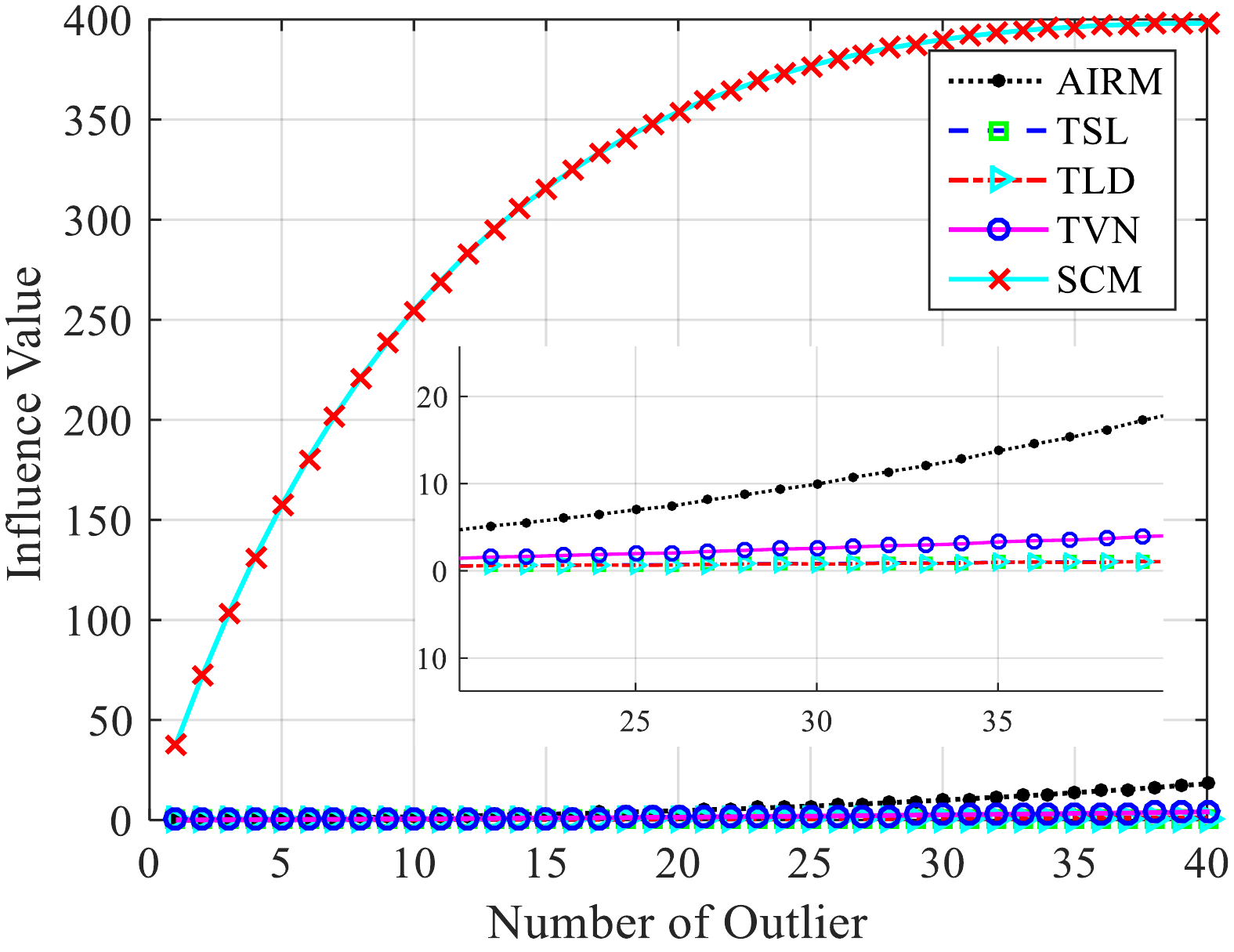}}
\caption{Influence values of TBD and AIRM means under different number of outliers}
\label{fig:Robustnes_Of_Num_Of_Outlier}
\end{figure}

\section{Simulation results}

In order to validate the  performances of the proposed TBD-MIG detector in nonhomogeneous clutter, numerical simulations are studied in this section. For comparison purposes, we also show the performances of AIRM-MIG detector and AMF. To decrease the computational load, the probability of false alarm ($P_{fa}$) is set to $10^{-3}$. The detection thresholds and probabilities of detection ($P_d$) are derived by $100/P_{fa}$ and $1000$ independent trials, respectively. For the sake of clarification, the AMF,  and the AIRM-MIG and TBD-MIG detectors  are repeated as follows:
\begin{equation*}
\begin{aligned}
T_{AMF} &= \frac{| \bm{x}^{\operatorname{H}}\bm{R}_{SCM}^{-1}\bm{p} |^2}{\bm{p}^{\operatorname{H}}\bm{R}_{SCM}^{-1}\bm{p}},  \\
T_{AIRM} &= \norm{ \operatorname{Log}(\bm{X}^{-1}\bm{R}_{AIRM}) },\\
T_{TSL} &= \frac{\norm{ \bm{X} - \bm{R}_{TSL} }^2}{\sqrt{1+\norm{ \bm{R}_{TSL} }^2}}, \\
T_{TLD} &= \frac{\ln\det(\bm{R}_{TLD}\bm{X}^{-1})+\operatorname{tr}(\bm{R}_{TLD}^{-1}\bm{X} - \bm{I})}{\sqrt{1+\norm{ \bm{R}_{TLD}^{-1} }^2}},  \\
T_{TVN} &= \frac{\operatorname{tr}\left(\bm{X}\operatorname{Log}\bm{X} - \bm{X}\operatorname{Log}\bm{R}_{TVN} - \bm{X} + \bm{R}_{TVN}\right)}{\sqrt{1+\norm{ \operatorname{Log}\bm{R}_{TVN} }^2}},
\end{aligned}
\end{equation*}
where $\bm{x}$ denotes the sample data, and $\bm{X}$ is an HPD matrix estimated by $\bm{x}$.

We consider two kinds of nonhomogeneous environments for target detection: one is the Gaussian clutter, while the other is the compound-Gaussian clutter where the texture component is assumed to be a Gamma distribution with the scale parameter $s$ and the shape parameter $v$. Two interferences are injected into clutter data. Values of the parameters used in the simulations, including the target, clutter and interference normalized Doppler frequencies $f_d,f_c,f_{in}$, are collected in Table \ref{Tab::paremater}.

\begin{table}[H]
  \centering
  \caption{Parameter values  used in the simulations}\label{Tab::paremater}
    \begin{tabular}{ccccccc}
		\toprule  
		$N$ &ICR (dB)&$f_c$&$f_d$&$f_{in}$&$s$&$v$ \\
		\midrule  
		8&20&0.1&0.2&0.22&1&3 \\
		\bottomrule  
    \end{tabular}
\end{table}

\begin{figure*}[h]
\centering
\subfigure[Gaussian  clutter, $K=8, N=8$] {\includegraphics[width=8.5cm,angle=0]{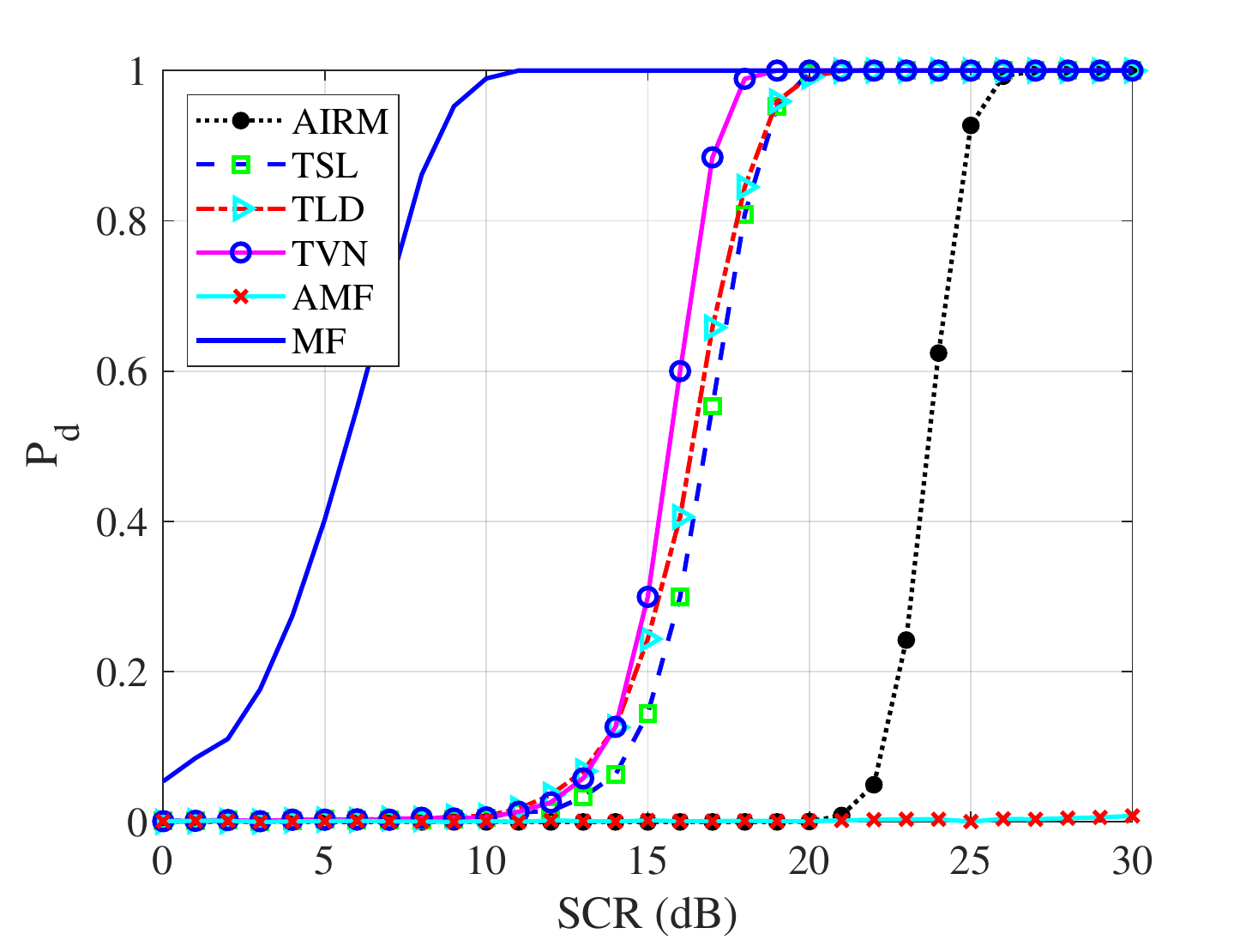}}
\subfigure[Non-Gaussian  clutter, $K=8, N=8$] {\includegraphics[width=8.5cm,angle=0]{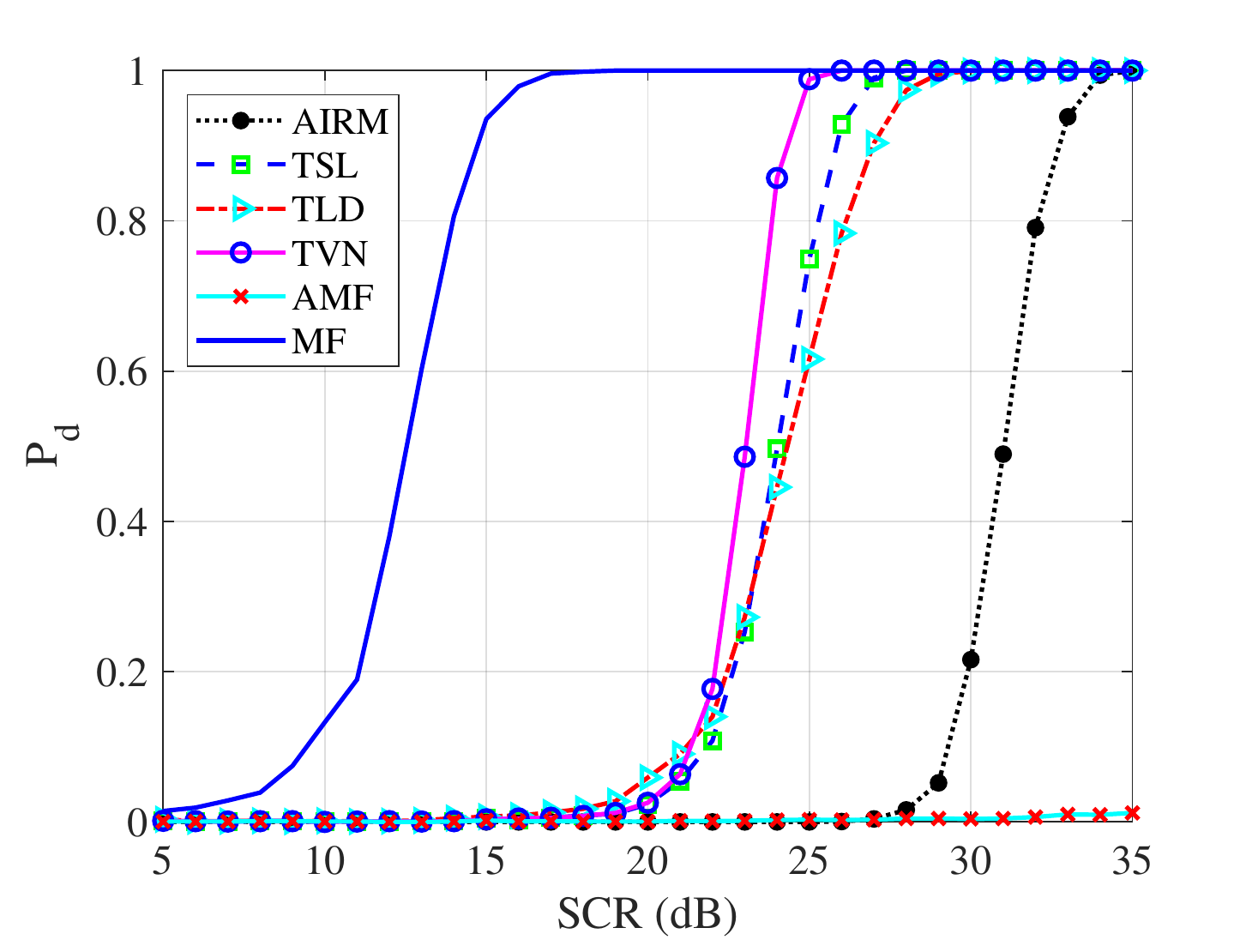}}
\subfigure[Gaussian clutter, $K=12, N=8$] {\includegraphics[width=8.5cm,angle=0]{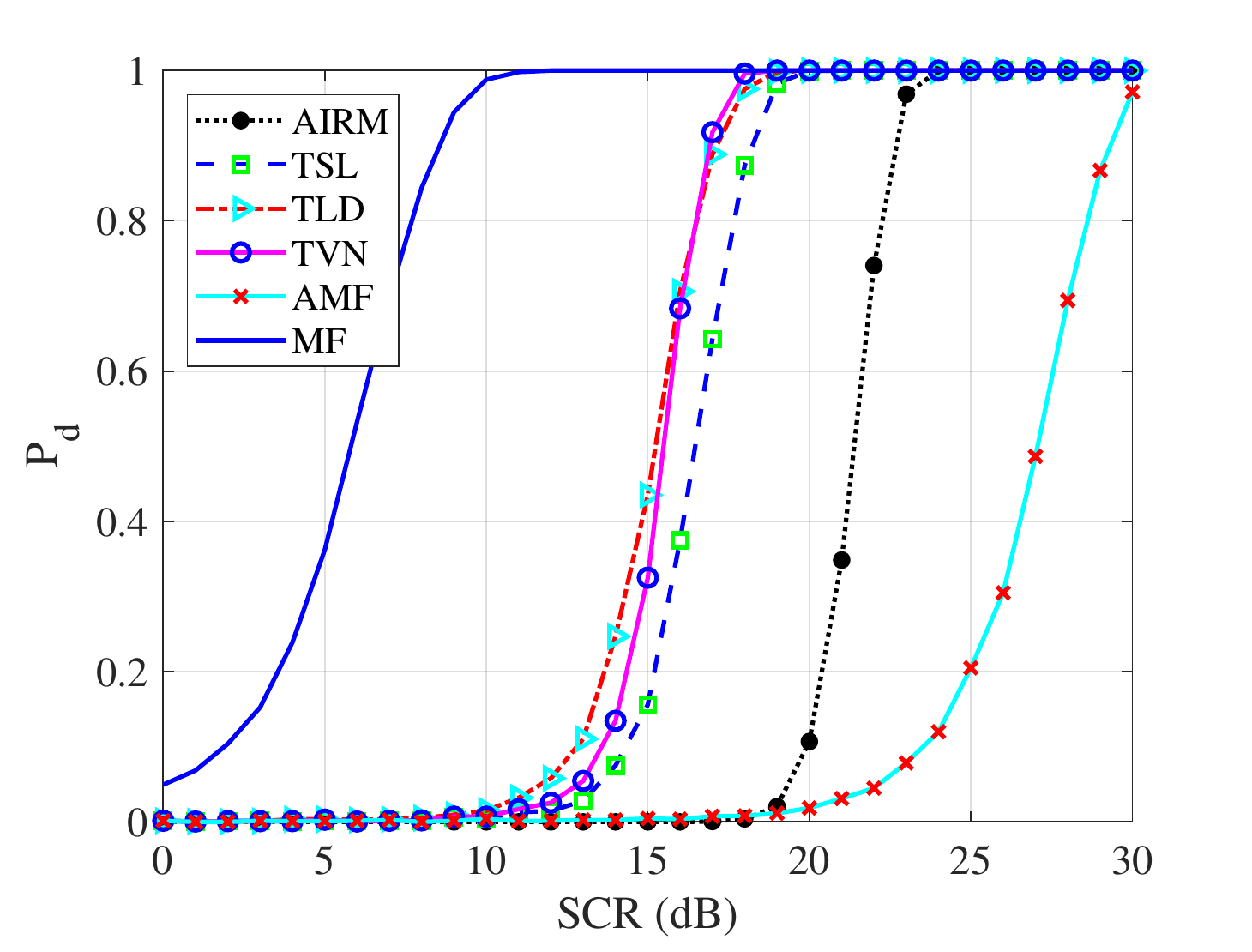}}
\subfigure[Non-Gaussian  clutter, $K=12, N=8$] {\includegraphics[width=8.5cm,angle=0]{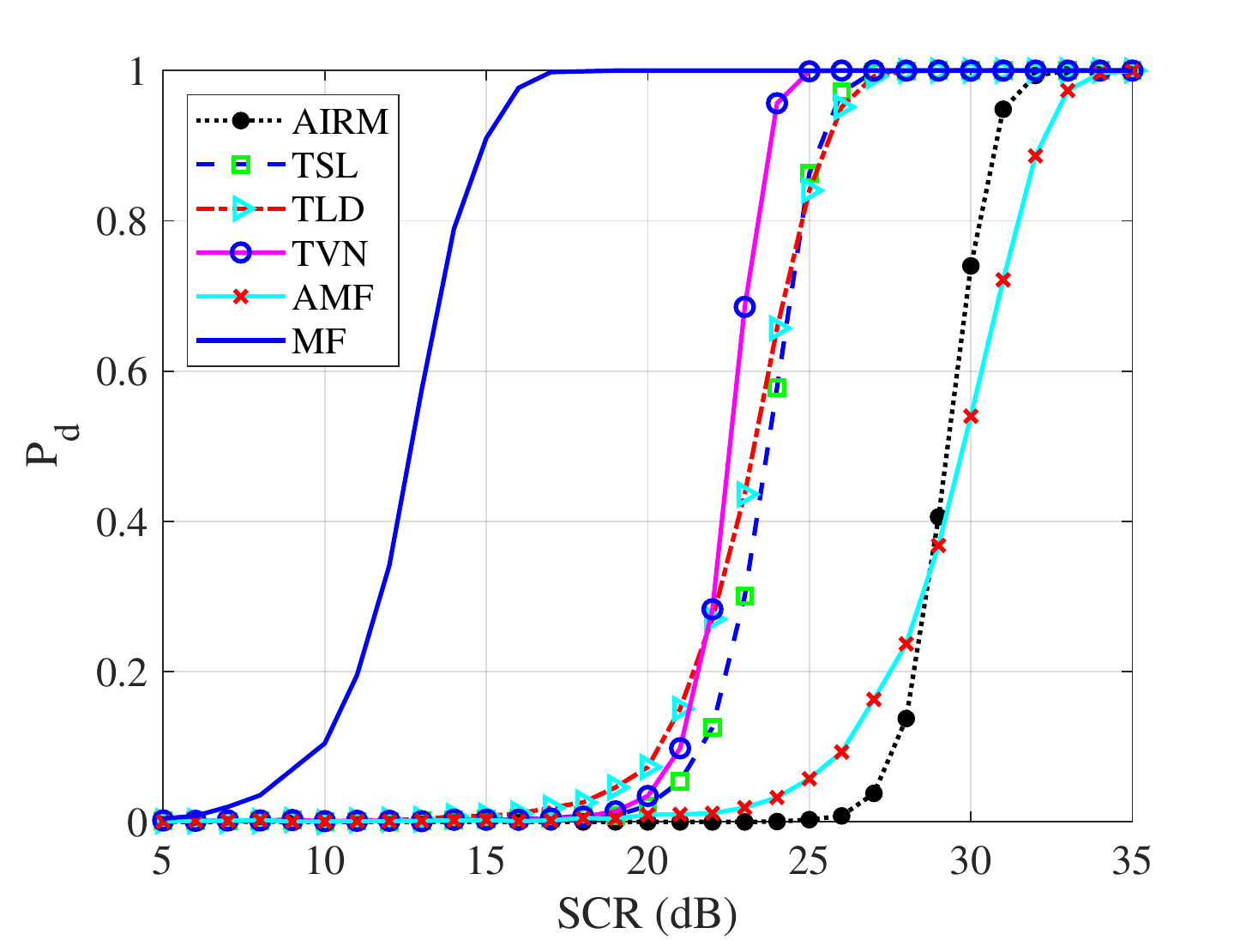}}
\subfigure[Gaussian  clutter, $K=16, N=8$] {\includegraphics[width=8.5cm,angle=0]{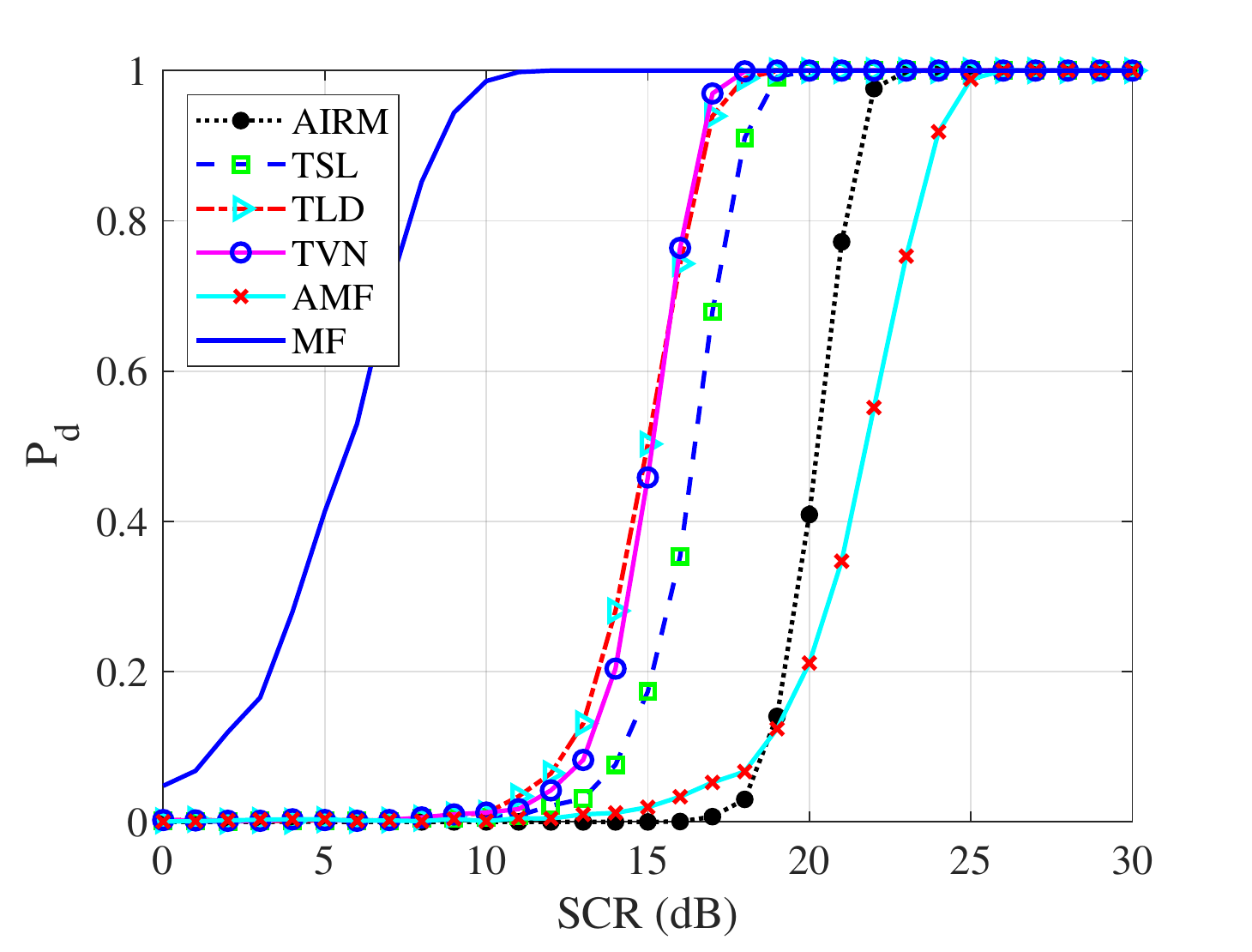}}
\subfigure[Non-Gaussian clutter, $K=16, N=8$] {\includegraphics[width=8.5cm,angle=0]{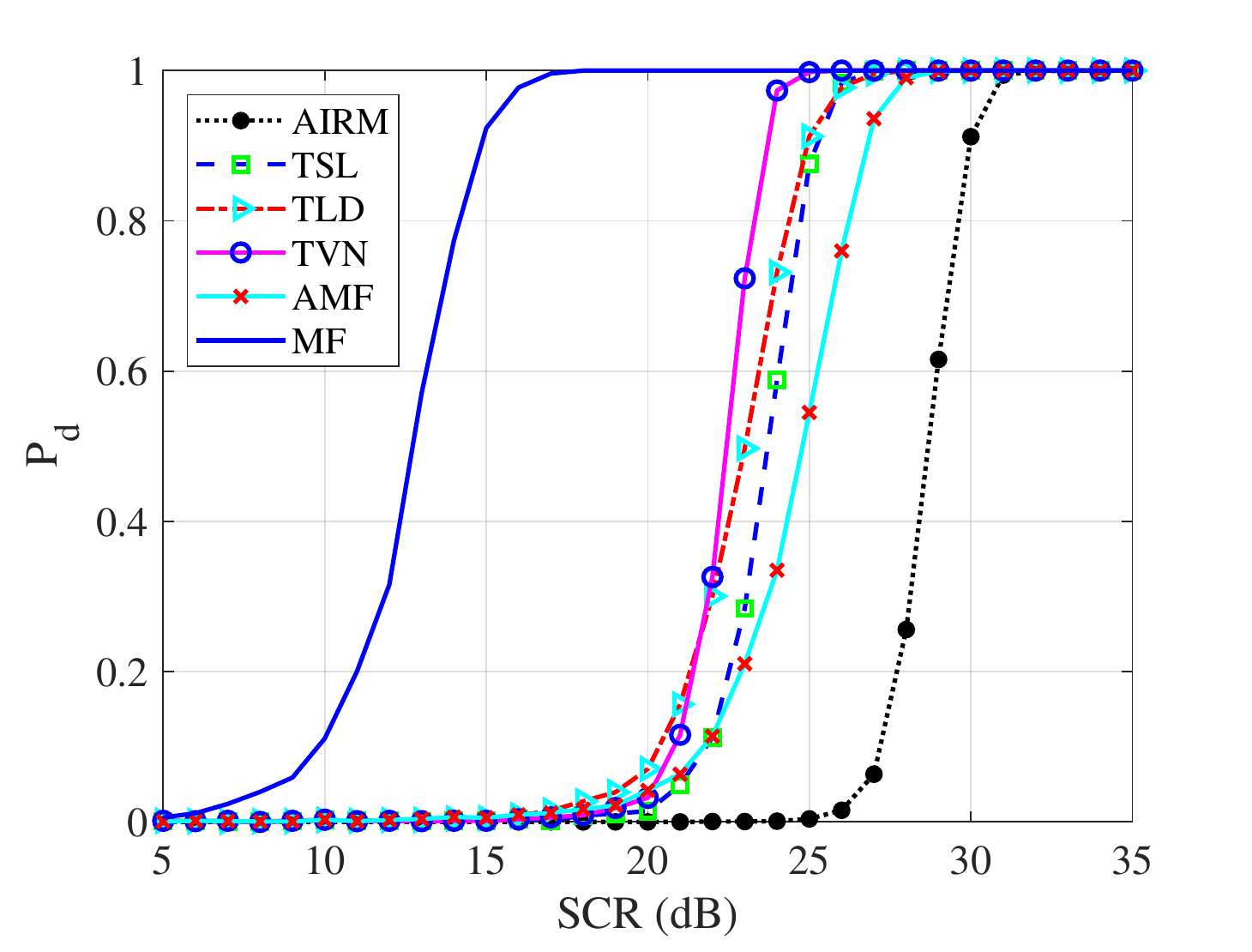}}
\caption{Probabilities of detection versus the signal to clutter ratio in Gaussian and non-Gaussian clutter for Toeplitz HPD structure, $P_{fa}=10^{-3}$.}
\label{fig:Toeplitz_PD_SCR}
\end{figure*}

\begin{figure*}[h]
\centering
\subfigure[Gaussian  clutter, $K=8, N=8$] {\includegraphics[width=8.5cm,angle=0]{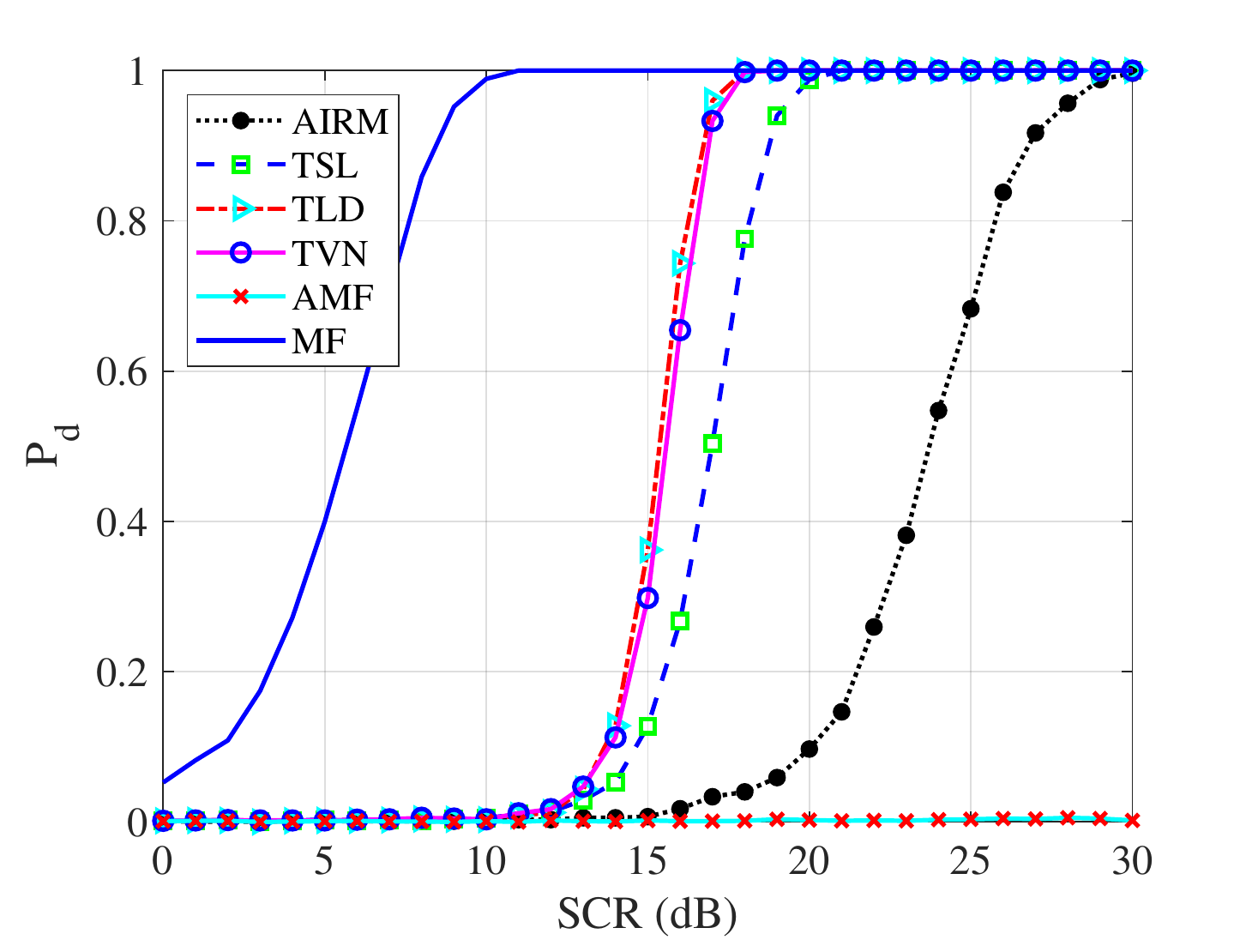}}
\subfigure[Non-Gaussian  clutter, $K=8, N=8$] {\includegraphics[width=8.5cm,angle=0]{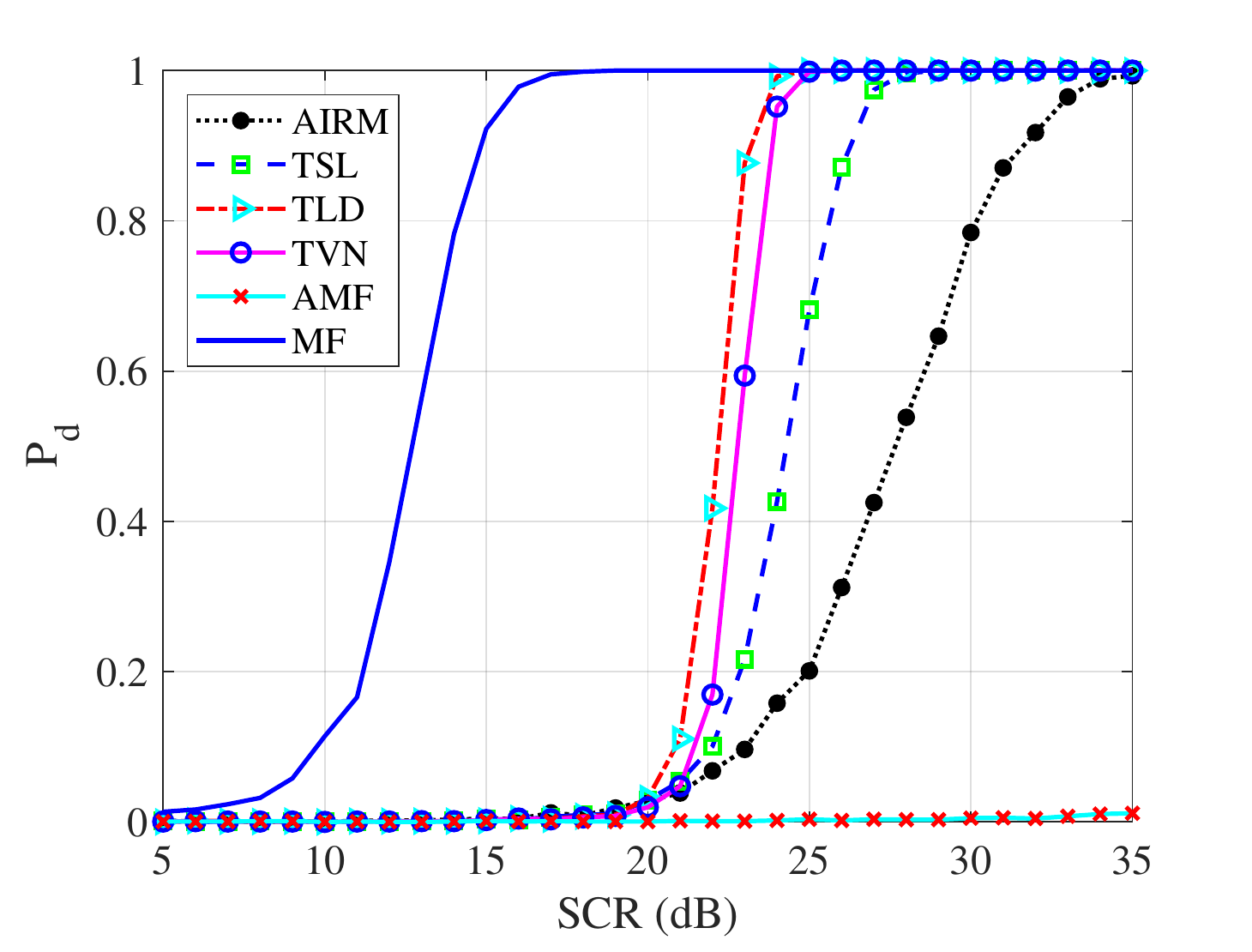}}
\subfigure[Gaussian  clutter, $K=12, N=8$] {\includegraphics[width=8.5cm,angle=0]{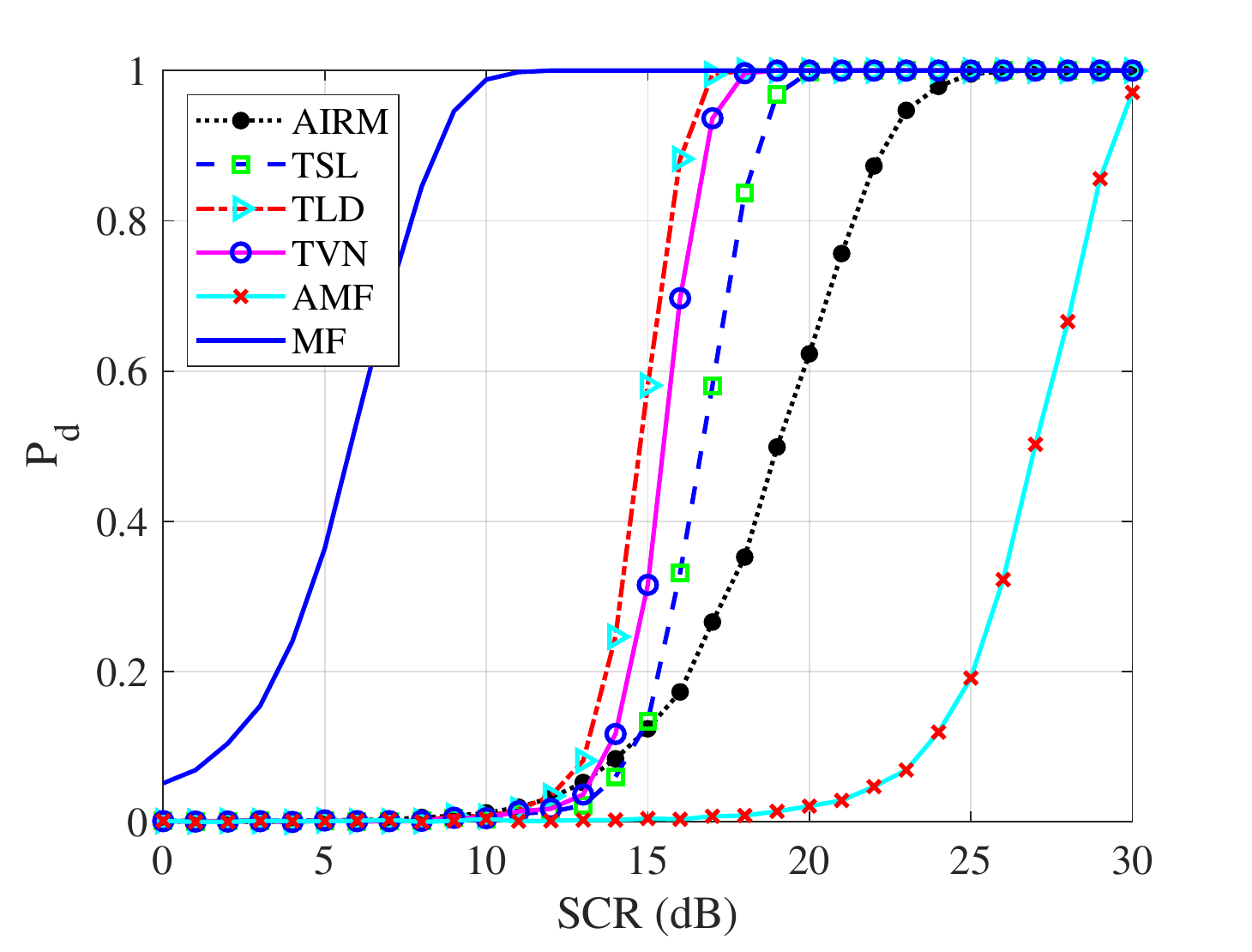}}
\subfigure[Non-Gaussian  clutter, $K=12, N=8$] {\includegraphics[width=8.5cm,angle=0]{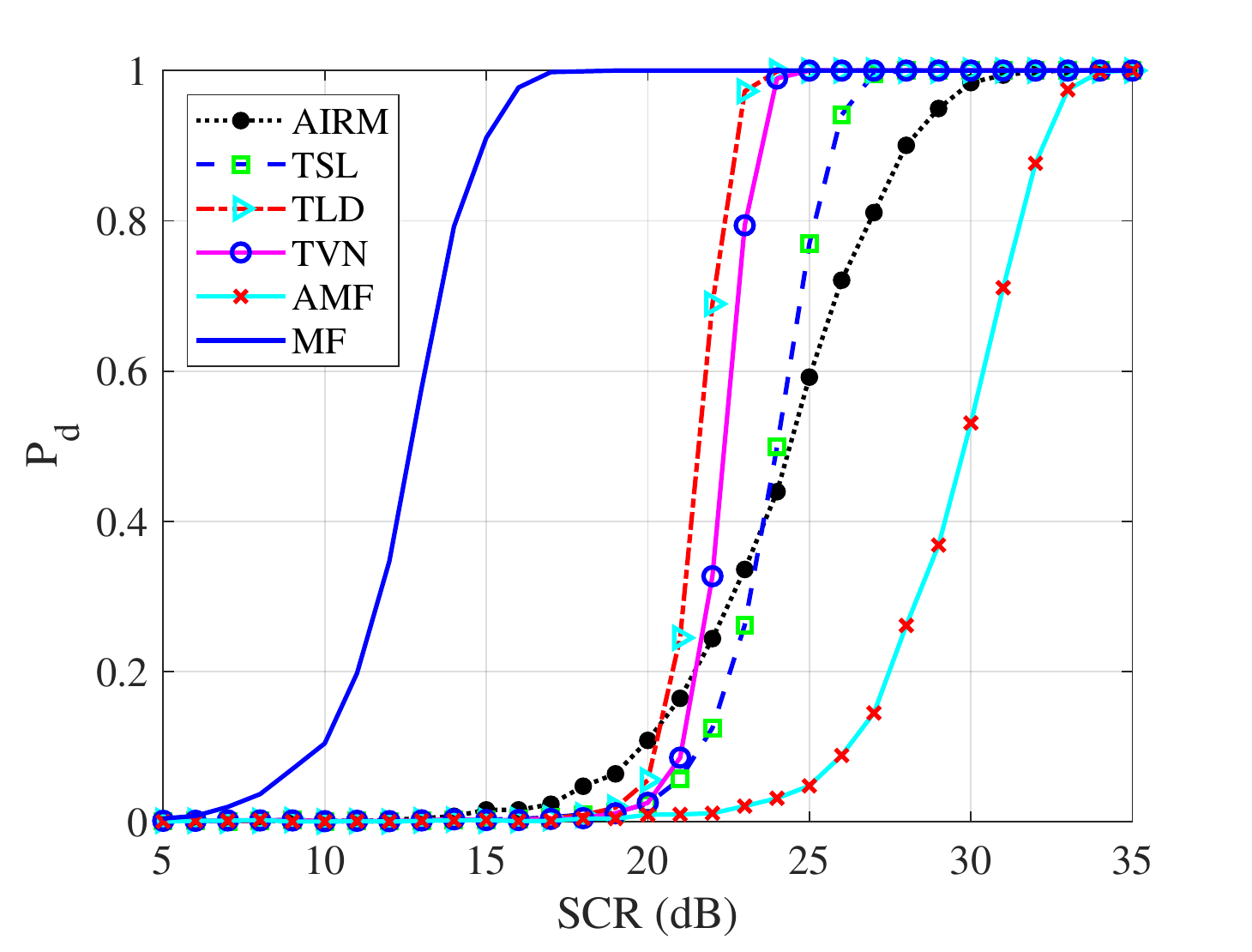}}
\subfigure[Gaussian clutter, $K=16, N=8$] {\includegraphics[width=8.5cm,angle=0]{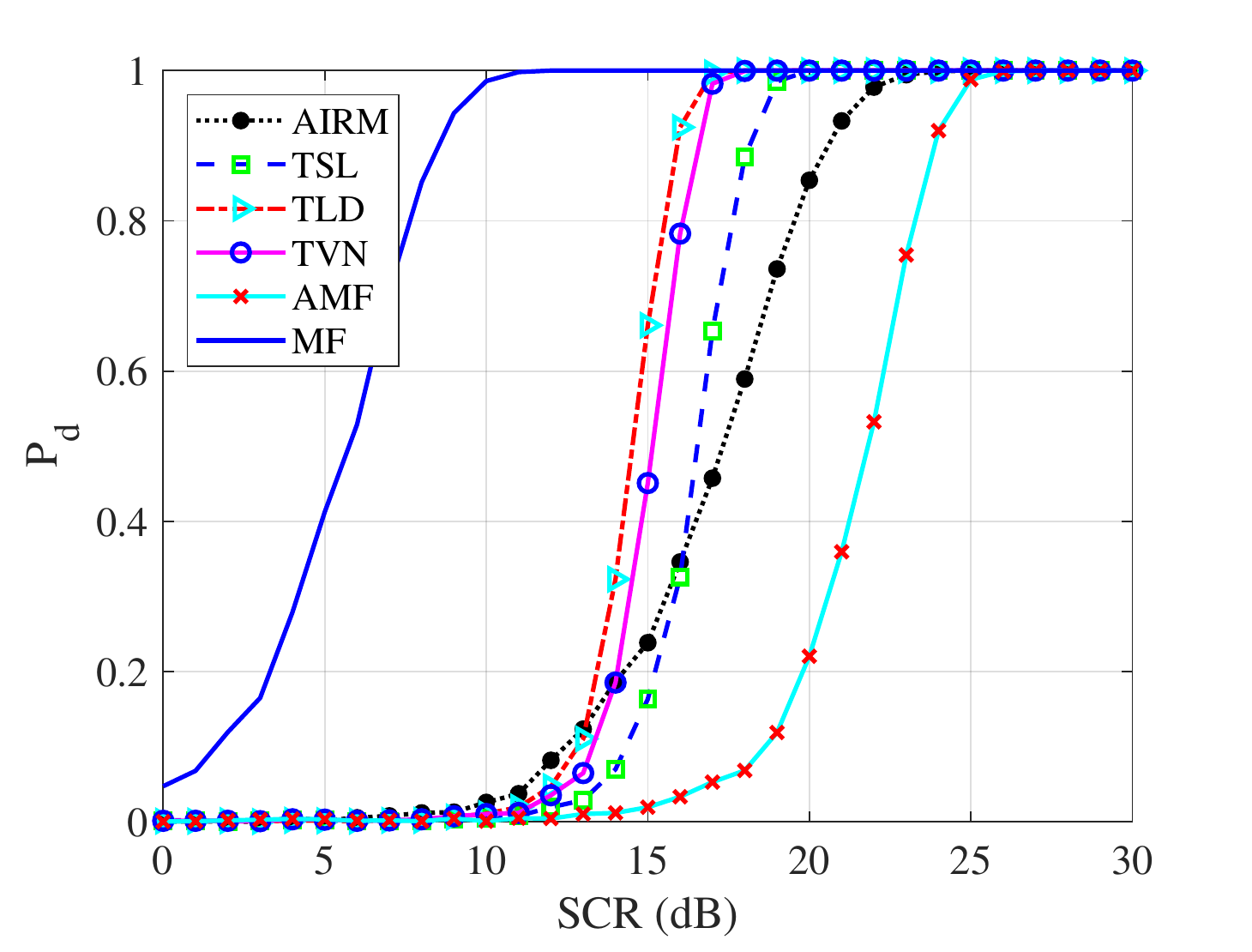}}
\subfigure[Non-Gaussian  clutter, $K=16, N=8$] {\includegraphics[width=8.5cm,angle=0]{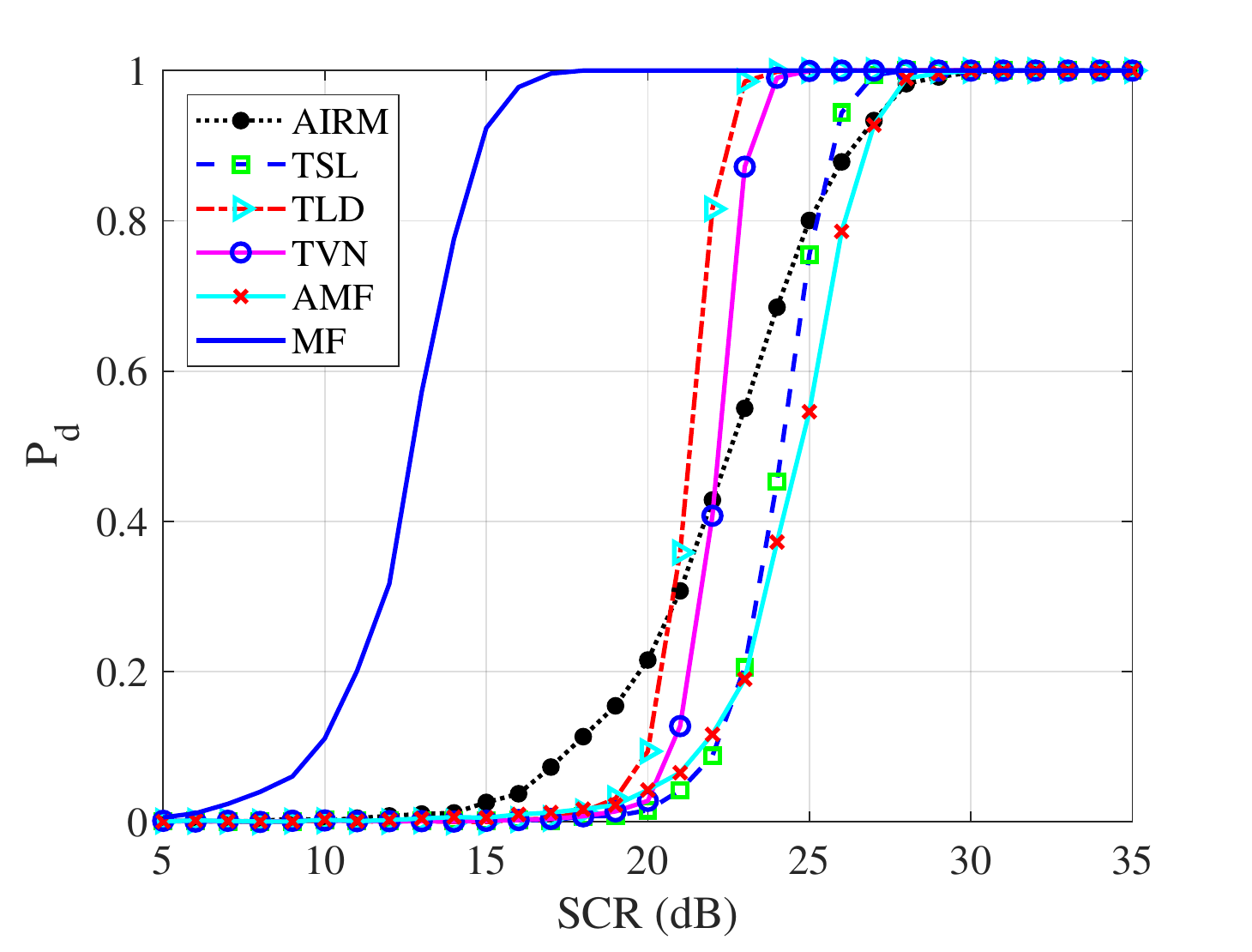}}
\caption{Probabilities of detection versus the signal to clutter ratio in Gaussian and non-Gaussian clutter for diagonal loading structure, $P_{fa}=10^{-3}$.}
\label{fig:Loading_PD_SCR}
\end{figure*}

In Fig. \ref{fig:Toeplitz_PD_SCR}, we compare the probabilities of detection of the proposed TBD-MIG detectors with the AIRM-MIG detector and the AMF in Gaussian and non-Gaussian clutters for Toeplitz HPD matrix case \eqref{eq:AR1}. Here, we also give the curves of the matched filter (MF) with the known covariance matrix as the performance benchmark for the AMF. Unlike the AMF, the optimal detection performances of the MIG detectors are difficult to determine, as the performance is closely related to the geometric measure used in the detector as well as the robustness of its corresponding geometric mean about outliers. Obviously, as the number of secondary data increases, performances of all  detectors improve. Particularly, all the considered detectors experience severe performance degradation in the non-Gaussian clutter with respect to the Gaussian clutter. When $K=N=8$, the AMF is invalid since the estimate error of the SCM is too larger. However, all the MIG detectors can still work in the case of $K=N$. The TBD-MIG detectors achieve significant performance advantage over the AIRM-MIG detector and the AMF in Gaussian and non-Gaussian clutters. The AIRM-MIG detector outperforms the AMF for $K=8,12$ in the Gaussian clutter, and for $K=8$ in the non-Gaussian clutter. Moreover, the AIRM-MIG detector can achieve performance improvement  for SCR bigger than $19$ dB in the case of $K=16$ in the Gaussian clutter and for SCR bigger than $29$ dB in the case of $K=12$ in the non-Gaussian clutter. The AMF has better performance than the AIRM-MIG detector for $K=16$ in the non-Gaussian clutter.

Fig. \ref{fig:Loading_PD_SCR} shows the performance comparison results of MIG detectors and the AMF in
 Gaussian and non-Gaussian clutters for the HPD matrix case \eqref{eq:AR2} obtained from diagonal loading.  Similar performance improvement can be seen as to the Toeplitz HPD matrix case when the number of secondary data increases. All  MIG detectors perform better than the AMF in Gaussian and non-Gaussian clutters for different $K$. The TBD-MIG detectors outperform the AIRM-MIG detector for $K=8, 12$ in the Gaussian clutter and for $K=8$ in the non-Gaussian clutter. Besides, the TBD-MIG detectors has better performance than the AIRM-MIG detector for SCR bigger than $16$ dB in the Gaussian clutter and for SCR bigger than $24$ dB in the non-Gaussian clutter.  In  the TBD-MIG detectors, the TLD-MIG detector has the best performance that is followed by the TVN-MIG detector.

\begin{figure}[H]
\centering
\subfigure[$Toeplitz \ Structure$] {\includegraphics[width=7.5cm,angle=0]{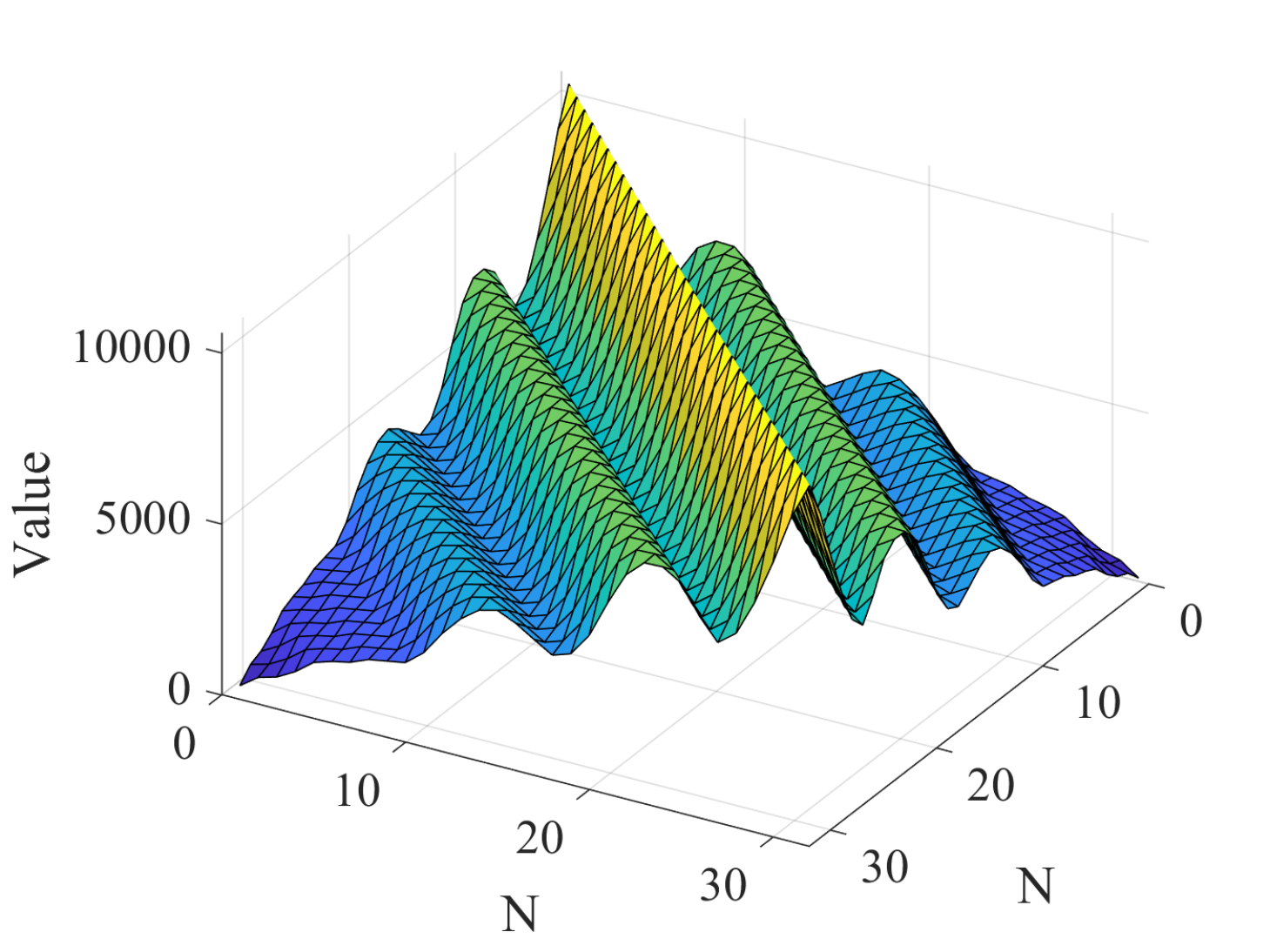}}
\subfigure[$Diagonal \ Structure$] {\includegraphics[width=7.5cm,angle=0]{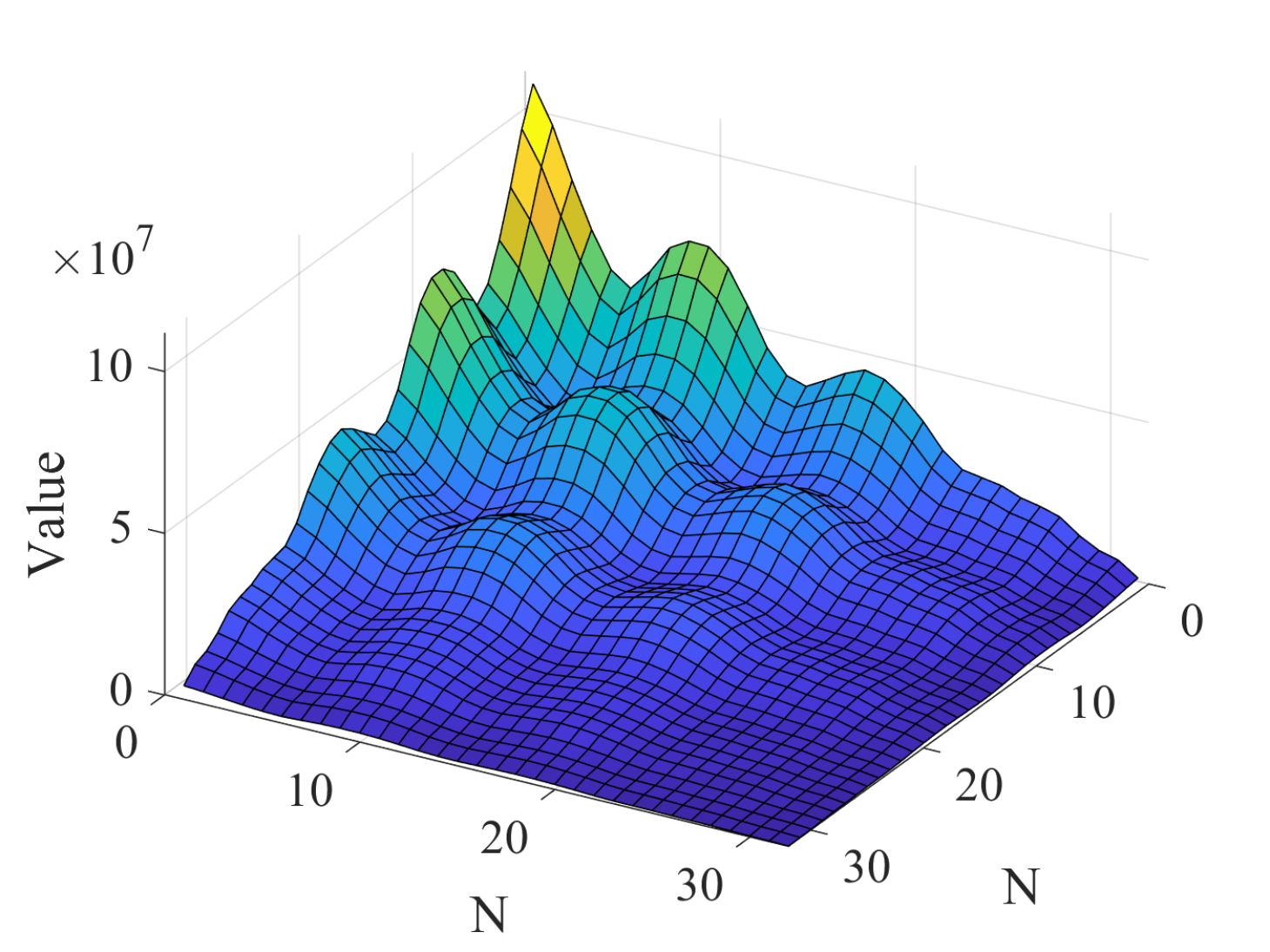}}
\caption{Energy distributions under different matrix structures.}
\label{fig:Matrix_Structure}
\end{figure}

Next we  examine the effect of different matrix structures on the detection performance, namely Toeplitz HPD matrices \eqref{eq:AR1} and HPD matrices obtained from diagonal loading \eqref{eq:AR2}. Firstly, we analyze  the energy distributions for these two matrix structures, shown in Fig. \ref{fig:Matrix_Structure}. It is observed that for the Toeplitz HPD matrix case, the energy is distributed parallel to the main diagonal and mainly distributed on the main diagonal. The farther the main diagonal is away, the less is the energy. However, for the diagonal loading HPD matrix case, the energy is mainly concentrated in a certain angle. The difference of their energy distributions is potentially one reason of their different detection performances.  In Fig. \ref{fig:Structure_PD_SCR}, we analyze  the performance of MIG detectors under different matrix structures. In the Gaussian clutter case, it is shown in Fig. \ref{fig:Structure_PD_SCR} (a) that the TLD-MIG and AIRM-MIG detectors with the diagonal loading structure can achieve several performance improvements over their counterparts with the Toeplitz structure, whereas both the TSL-MIG and TVN-MIG detectors have similar performance with the diagonal loading and Toeplitz structures. In the non-Gaussian clutter case, Fig. \ref{fig:Structure_PD_SCR} (b) shows that except for the similar performance of TSL-MIG detector, the performances of all MIG detectors with the diagonal loading structure are better to that of the detectors with the Toeplitz structure.

\begin{figure}[H]
\centering
\subfigure[Gaussian  clutter] {\includegraphics[width=7.5cm,angle=0]{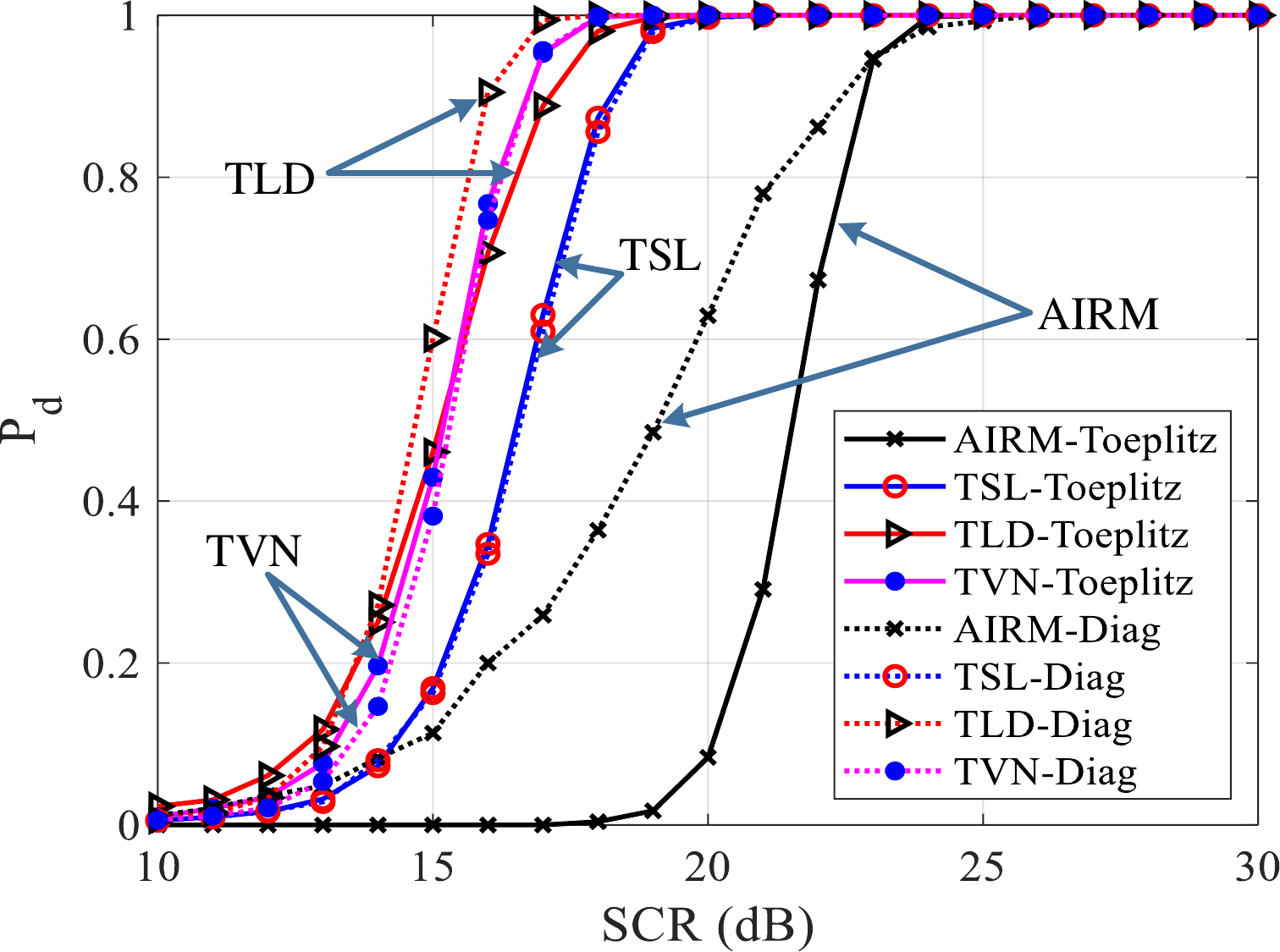}}
\subfigure[Non-Gaussian \ clutter] {\includegraphics[width=7.5cm,angle=0]{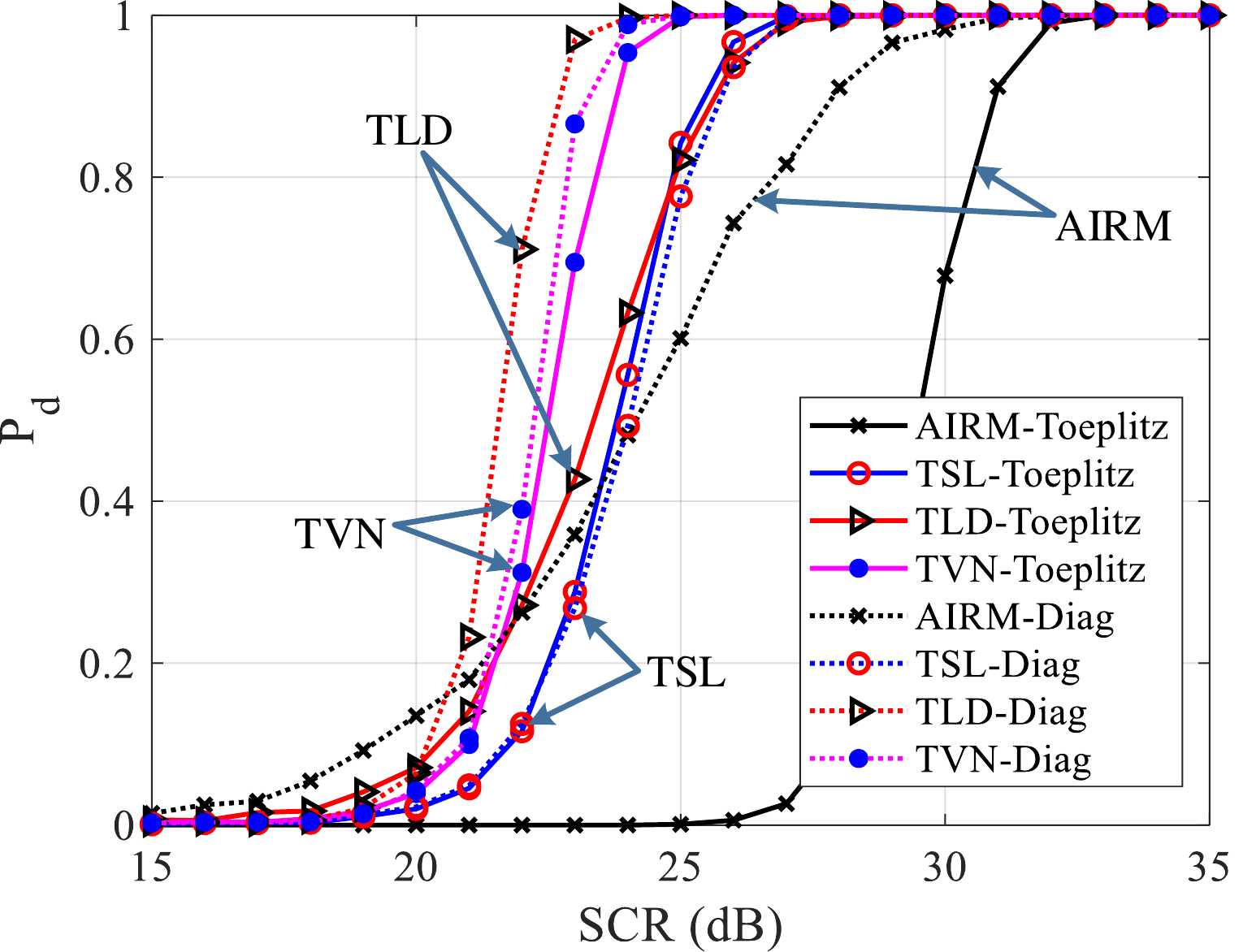}}
\caption{Probabilities of detection versus the signal to clutter ratio under different matrix structures.}
\label{fig:Structure_PD_SCR}
\end{figure}

\section{Conclusions}

In this paper, we  proposed a TBD-MIG detector to investigate the problem of target detection in nonhomogeneous clutter. The sample data has been assumed to be modeled as HPD matrices, which is used as the secondary data to estimate the CCM by the TBD mean. We then reformulated the problem of signal detection as discriminating two points on the HPD matrix manifold. Three TBD-MIG detectors, referred to as the TSL-MIG, TLD-MIG and TVN-MIG detectors, were introduced. Influence functions related to  geometric means with respect to different divergences were defined and calculated in closed-form, providing convenience of a theoretical analysis of the robustness to outliers. Interestingly, the TBD mean is upper bounded about the interference energy. Numerical simulations shown that the TBD-MIG detectors outperform the AIRM-MIG detector and the AMF in homogeneous clutter. Energy distributions of two matrix structures, i.e., Toeplitz HPD matrices and HPD matrices obtained from diagonal loading, were analyzed, and a comparison of their influences on detection performance was also conducted.

 From the theoretical aspect, it would be interesting to study the Riemannian-geometric structures induced from the proposed total divergences as well as their relations with that of the well-known AIRM and the Log-Euclidean geometry of HPD matrix manifolds \cite{AFPA2007}.
Possible future research in applications will concern  the target detection in real datasets using the TBD-MIG detectors and the extension of MIG detectors to the structured covariance interference, e.g., persymmetric covariance matrices \cite{COP2016,ZWZN2020}.

\appendices
\section{Proof of Proposition \ref{prop:tvn}}
\label{appen:A}

The following two lemmas are used in the proof of Proposition \ref{prop:tvn}.

\begin{lem}[\cite{Hig2008,Moa2005}] \label{lem:aa}
Suppose $\bm{X}$ is an invertible matrix that does not have eigenvalues in the closed negative real line and  denote $\operatorname{Log}\bm{X}$ its principal logarithm. The matrix $\bm{X}$ satisfies the following properties.
\begin{itemize}
\item[(i)] Both $\bm{X}$ and $\operatorname{Log}\bm{X}$ commute with $[(\bm{X}-\bm{I})s+\bm{I}]^{-1}$ for any real number $s$.
\item[(ii)]   The following identity holds that
\begin{equation*}
\begin{aligned}
\int_0^1[(\bm{X}&-\bm{I})s+\bm{I}]^{-2}\operatorname{d}\!s\\
&=(\bm{I}-\bm{X})^{-1}[(\bm{X}-\bm{I})s+\bm{I}]^{-1}\Big|_{s=0}^1\\
&=\bm{X}^{-1}.
\end{aligned}
\end{equation*}
\end{itemize}
\end{lem}

\begin{lem}[\cite{Moa2005}] \label{lem:ln}
\begin{itemize}
\item[(i)] For an arbitrary matrix $\bm{B}(s)$ with $s\in\mathbb{R}$ and arbitrary real numbers $a<b$, the following commutative property holds:
\begin{equation*}
\operatorname{tr}\left(\int_a^b\bm{B}(s)\operatorname{d}\!s\right)=\int_a^b\operatorname{tr}\left(\bm{B}(s)\right)\operatorname{d}\!s.
\end{equation*}
\item[(ii)] Suppose $\bm{A}(\varepsilon)$ is an invertible matrix which does not have eigenvalues lying in the closed real line.
Then, we have
\begin{equation*}
\begin{aligned}
\frac{\operatorname{d}}{\operatorname{d}\!\varepsilon}\operatorname{Log}\bm{A}(\varepsilon)=\int_0^1&
[(\bm{A}(\varepsilon)-\bm{I})s+\bm{I}]^{-1} \\
&\times \frac{\operatorname{d}\!\bm{A}(\varepsilon)}{\operatorname{d}\!\varepsilon}[(\bm{A}(\varepsilon)-\bm{I})s+\bm{I}]^{-1}\operatorname{d}\!s.
\end{aligned}
\end{equation*}
\end{itemize}
\end{lem}

{\it Proof of Proposition \ref{prop:tvn}.}
Denote $\bm{A}(\varepsilon):=\bm{X}+\varepsilon\bm{Y}$, which is assumed to have no eigenvalues lying in the negative real line. Gradient of the function $F$ is calculated as
\begin{equation*}
\begin{aligned}
\langle \nabla F&(\bm{X}),\bm{Y}\rangle:=\frac{\operatorname{d}}{\operatorname{d}\!\varepsilon}\Big|_{\varepsilon=0} F(\bm{X}+\varepsilon \bm{Y})\\
&=\frac{\operatorname{d}}{\operatorname{d}\!\varepsilon}\Big|_{\varepsilon=0}\operatorname{tr}\left(\bm{A}(\varepsilon)\operatorname{Log}\bm{A}(\varepsilon)-\bm{A}(\varepsilon)\right)\\
&=\operatorname{tr}\left(\frac{\operatorname{d}\!\bm{A}(\varepsilon)}{\operatorname{d}\!\varepsilon}\operatorname{Log}\bm{A}(\varepsilon)\right)\Big|_{\varepsilon=0}\\
&~~~~+\operatorname{tr}\left(\bm{A}(\varepsilon)\frac{\operatorname{d}}{\operatorname{d}\!\varepsilon}\operatorname{Log}\bm{A}(\varepsilon)\right)\Big|_{\varepsilon=0}
-\operatorname{tr}(\bm{Y})\\
&=\operatorname{tr}\left(\bm{Y}\operatorname{Log}\bm{X}-\bm{Y}\right)+\operatorname{tr}\left(\bm{A}(\varepsilon)\frac{\operatorname{d}}{\operatorname{d}\!\varepsilon}\operatorname{Log}\bm{A}(\varepsilon)\right)\Big|_{\varepsilon=0}.
\end{aligned}
\end{equation*}
 What left is to compute the last differentiation term. By using Lemmas \ref{lem:aa} and \ref{lem:ln}, we have
 \begin{equation*}
 \begin{aligned}
 \operatorname{tr}&\left(\bm{A}(\varepsilon)\frac{\operatorname{d}}{\operatorname{d}\!\varepsilon}\operatorname{Log}\bm{A}(\varepsilon)\right)
 =\operatorname{tr}\left(\int_0^1\bm{A}(\varepsilon)\left[(\bm{A}(\varepsilon)-\bm{I})s+\bm{I}\right]^{-1}\right.\\
 &~~~~~~~~~~~~~~~~\left.\frac{\operatorname{d}\!\bm{A}(\varepsilon)}{\operatorname{d}\!\varepsilon}\left[(\bm{A}(\varepsilon)-\bm{I})s+\bm{I}\right]^{-1}\operatorname{d}\!s\right)\\
 &=\int_0^1\operatorname{tr}\Big(\bm{A}(\varepsilon)\left[(\bm{A}(\varepsilon)-\bm{I})s+\bm{I}\right]^{-1}\\
 &~~~~~~~~~~~~~~~~\frac{\operatorname{d}\!\bm{A}(\varepsilon)}{\operatorname{d}\!\varepsilon}\left[(\bm{A}(\varepsilon)-\bm{I})s+\bm{I}\right]^{-1}\Big)\operatorname{d}\!s\\
 &=\int_0^1\operatorname{tr}\left(\left[(\bm{A}(\varepsilon)-\bm{I})s+\bm{I}\right]^{-2}\bm{A}(\varepsilon)\frac{\operatorname{d}\!\bm{A}(\varepsilon)}{\operatorname{d}\!\varepsilon}\right)\operatorname{d}\!s\\
  &=\operatorname{tr}\left(\int_0^1\left[(\bm{A}(\varepsilon)-\bm{I})s+\bm{I}\right]^{-2}\operatorname{d}\!s~\bm{A}(\varepsilon)\frac{\operatorname{d}\!\bm{A}(\varepsilon)}{\operatorname{d}\!\varepsilon}\right)\\
 &=\operatorname{tr}\left((\bm{I}-\bm{A}(\varepsilon))^{-1}\left[(\bm{A}(\varepsilon)-\bm{I})s+\bm{I}\right]\Big|_{s=0}^1\bm{A}(\varepsilon)\frac{\operatorname{d}\!\bm{A}(\varepsilon)}{\operatorname{d}\!\varepsilon}\right)\\
  &=\operatorname{tr}\left(\frac{\operatorname{d}\!\bm{A}(\varepsilon)}{\operatorname{d}\!\varepsilon}\right),
 \end{aligned}
 \end{equation*}
 that equals to $\operatorname{tr}(\bm{Y})$ for the  matrix $\bm{A}(\varepsilon)=\bm{X}+\varepsilon\bm{Y}$. Therefore, we have
\begin{equation*}
\langle \nabla F(\bm{X}),\bm{Y}\rangle=\operatorname{tr}\left(\bm{Y}\operatorname{Log}\bm{X}\right)=\langle \left(\operatorname{Log}\bm{X}\right)^{\operatorname{H}},\bm{Y}\rangle,
\end{equation*}
and consequently $\nabla F(\bm{X})=(\operatorname{Log} \bm{X})^{\operatorname{H}}$ and  norm of the gradient  is
\begin{equation*}
\begin{aligned}
\norm{\nabla F(\bm{X})} &=\norm{ (\operatorname{Log} \bm{X})^{\operatorname{H}}}.
\end{aligned}
\end{equation*}
Consequently, the Bregman divergence is given by
\begin{equation*}
\operatorname{B}_F(\bm{X},\bm{Y})= \operatorname{tr}\left( \bm{X}(\operatorname{Log} \bm{X} - \operatorname{Log} \bm{Y}) - \bm{X} + \bm{Y} \right)
\end{equation*}
and hence we obtain the TVN as
\begin{equation*}
\delta_F(\bm{X},\bm{Y})=\frac{\operatorname{tr}\left( \bm{X}(\operatorname{Log} \bm{X} - \operatorname{Log} \bm{Y} ) - \bm{X} + \bm{Y} \right)}{\sqrt{1+\norm{ (\operatorname{Log} \bm{Y})^{\operatorname{H}}}^2}}.
\end{equation*}

\section{Proof of Proposition \ref{prop:TBDae}}
\label{appen:B}
Define $G(\bm{X})$ as the objective function to be minimized for the enlarged $m+n$ HPD matrices, namely
\begin{equation*}
\begin{aligned}
G(\bm{X}) :&= (1-\varepsilon)\frac{1}{m}\sum_{i=1}^{m}\norm{\operatorname{Log}(\bm{X}_i^{-1}\bm{X})}^2 \\
&~~~~+ \varepsilon\frac{1}{n}\sum_{j=1}^{n}\norm{\operatorname{Log}(\bm{P}_j^{-1}\bm{X})}^2.
\end{aligned}
\end{equation*}
The gradient of the norm function with respect to the AIRM was shown to be \cite{Moa2005}
\begin{equation*}
\nabla \norm{\operatorname{Log}\left(\bm{X}_i^{-1}\bm{X}\right)}^2=2\bm{X}\operatorname{Log}\left(\bm{X}_i^{-1}\bm{X}\right).
\end{equation*}
Consequently, we have
\begin{equation*}
\begin{aligned}
\nabla G(\bm{X})&= 2(1-\varepsilon)\frac{1}{m}\sum_{i=1}^{m}\bm{X}\operatorname{Log}(\bm{X}_i^{-1}\bm{X})\\
&~~~~+ 2\varepsilon\frac{1}{n}\sum_{j=1}^{n}\bm{X}\operatorname{Log}(\bm{P}_j^{-1}\bm{X}).
\end{aligned}
\end{equation*}

As $\bm{\widehat{X}}=\overline{\bm{X}}+\varepsilon \bm{H}+O(\varepsilon^2)$ is the mean of $m$ HPD matrices $\{ \bm{X}_1, \bm{X}_2, \ldots, \bm{X}_m \}$ and $n$ outliers $\{\bm{P}_1, \bm{P}_2, \ldots, \bm{P}_n\}$, then $\nabla G(\bm{\widehat{X}})=\bm{0}$, namely
\begin{equation*}\label{eq:AIRM1}
(1-\varepsilon)\frac{1}{m}\sum_{i=1}^{m}\operatorname{Log}(\bm{X}_i^{-1}\bm{\widehat{X}}) + \varepsilon\frac{1}{n}\sum_{j=1}^{n}\operatorname{Log}(\bm{P}_j^{-1}\bm{\widehat{X}})=\bm{0}.
\end{equation*}
To obtain the linear term about $\varepsilon$, we can simply differentiate the equality above and then set $\varepsilon$ to be zero. By doing so, we obtain
\begin{equation}\label{eq:lp1}
\frac{1}{n}\sum_{j=1}^n\operatorname{Log}\left(\bm{P}_j^{-1}\overline{\bm{X}}\right)+\frac{1}{m}\sum_{i=1}^m\frac{\operatorname{d}}{\operatorname{d}\!\varepsilon}\Big|_{\varepsilon=0} \operatorname{Log}\left(\bm{X}_i^{-1}\widehat{\bm{X}}\right)=\bm{0}.
\end{equation}
The condition that $\overline{\bm{X}}$ is the mean of $m$ HPD matrices $\{ \bm{X}_1, \bm{X}_2, \ldots, \bm{X}_m \}$ is applied, i.e.,
$$
\frac{1}{m}\sum_{i=1}^{m}\operatorname{Log}(\bm{X}_i^{-1}\bm{\overline{X}}) = \bm{0}.
$$

Taking trace of the identity \eqref{eq:lp1} and using Lemmas \ref{lem:aa} and \ref{lem:ln}, a similar calculation as the proof of Proposition \ref{prop:tvn} (see Appendix \ref{appen:A}) yields
\begin{equation*}
\operatorname{tr}\left(\overline{\bm{X}}^{-1}\bm{H}+\frac{1}{n}\sum_{j=1}^n\operatorname{Log}\left(\bm{P}_j^{-1}\overline{\bm{X}}\right)\right)=0.
\end{equation*}
This can be written using the metric \eqref{eq:glmetric} as well. 
Assuming the arbitrarity of  $\overline{\bm{X}}$, we can choose
\begin{equation*}
\bm{H}= - \frac{1}{n}\sum_{j=1}^{n}\frac{\bm{\overline{X}}  \operatorname{Log}(\bm{P}_j^{-1}\bm{\overline{X}})+\operatorname{Log}(\bm{\overline{X}} \bm{P}_j^{-1})\bm{\overline{X}}  }{2}.
\end{equation*}
This finishes the proof.

\bibliographystyle{IEEEtran}

\bibliography{mybibfile}

\end{document}